\documentclass[letterpaper, 11pt, leqno]{article}

\usepackage{amssymb}
\usepackage{amsmath}
\usepackage{amsfonts}
\usepackage{amsthm}
\usepackage{mathtools}
\usepackage{graphicx}
\usepackage{enumerate}
\usepackage[margin=1.15in]{geometry}
\usepackage{setspace}
\usepackage[applemac]{inputenc}
\usepackage{dsfont}
\usepackage{natbib}
\usepackage{mathrsfs}
\usepackage{color}
\usepackage{comment}
\usepackage{latexsym}
\usepackage{cancel}
\usepackage{adjustbox}

\usepackage{arydshln}
\usepackage{booktabs}

\usepackage{geometry} 
\usepackage{xcolor, graphicx} 
\usepackage{tikz}
\usetikzlibrary{trees}

\usepackage{hyperref}
\usepackage{nameref}
\hypersetup{
breaklinks = true,
  colorlinks   = true, 
  urlcolor     = {blue!50!black}, 
  linkcolor    = {violet!90!black}, 
  citecolor   = {gray}
}
\usepackage{cleveref}

\newtheorem{thm}{Theorem}
\newtheorem{prps}{Proposition}

\theoremstyle{definition}
\newtheorem{axm}{\sc Axiom}
\newtheorem{defn}{Definition}
\newtheorem{cor}{Corollary}

\newtheorem{exm}{Example}

\title{Inertial Updating with General Information\thanks{Dominiak: Virginia Tech (dominiak@vt.edu); Kovach: Purdue University (mlkovach@purdue.edu); Tserenjigmid: UC Santa Cruz (gtserenj@ucsc.edu ). This paper was previously circulated as ``Minimum Distance Updating with General Information." We are very grateful to Yeon-Koo Che, David Freeman, Paolo Ghirardato, Faruk Gul, Emir Kamenica, Edi Karni, Michihiro Kandori, Daniel Martin, Yusufcan Masatlioglu, Giacomo Lanzani, Kota Saito, and Jakub Steiner for valuable comments and discussions, as well as the seminar participants at Caltech, CUHK-HKU-HKUST Joint Theory Seminar; Hitotsubashi; NUS; HKU; HKUST; University of Maryland; University of Michigan; University of Tokyo; UC Santa Cruz; UC Riverside; UC Davis; UC Berkeley; Royal Holloway (University of London); Dauphine University; ETH Z\"urich; Western University, along with the audiences of the Conference on Unawareness and Unintended Consequences (2021), RUD 2021 (University of Minnesota); 20th SAET Conference (Seoul); D-TEA (PSE); ESA - North American meeting 2024 (Columbus); BRIC 2022 (CERGE-EI); International Conference on Social Choice and Voting Theory (Virginia Tech), the Workshop in Experiments, Revealed Preferences, and Decisions (WiERD, 2020), Barcelona Summer Forum 2025, and Econometric Society World Congress 2025 (Seoul).}}
\vspace{5 mm}
\author{\centering Adam Dominiak  \and Matthew Kovach \and Gerelt Tserenjigmid}

\begin{document}

\maketitle

\noindent{\textbf{Abstract:} \onehalfspacing We study belief revision when a data-generating process (DGP) is unknown, and information is represented by a set of probability distributions, or \emph{general information}. General information extends the standard event notion while including qualitative information (\emph{A is more likely than B}), interval information (\emph{A has a ten-to-twenty percent chance}), and more.  We behaviorally characterize \emph{Inertial Updating}: the decision maker's posterior is of minimal subjective distance from her prior, given the information constraint. Further, we introduce and characterize a notion of Bayesian updating for general information and show that \emph{Bayesian agents may disagree}. We also behaviorally characterize \emph{f-divergences}, the class of distances consistent with Bayesian updating.

}

\vspace{5 mm}


\vspace{3 mm}

\vspace{30 mm}

\newpage

\section{Introduction}
\onehalfspacing

How decision makers revise their beliefs after receiving information is a foundational problem in economics and game theory. In the standard Bayesian framework, belief revision requires the decision maker to have a sufficiently precise model of the data-generating process (DGP): she must be able to interpret messages or signals through a joint distribution over payoff-relevant states and possible observations. In many economically important environments, however, this is implausible. Information is often qualitative, unexpected, or generated by an opaque mechanism, such as a medical study or a black-box algorithmic recommendation. In such settings, the decision maker may understand the content of the information without possessing a credible model of how that information is generated.

Because the meaning of signals is typically derived from the DM's understanding of the DGP, a new notion of information is required when the DGP is unknown.  Accordingly, we consider statements that have objective, qualitative content, which we refer to as general information. Further, such informational statements are incredibly common in daily life.   To illustrate, consider the following statements:
\begin{itemize}

\item[(i)] \emph{Among similar patients, $80\%$ experienced at least partial improvement.}

\item[(ii)]	\emph{The chance of rain tomorrow is between 60\% and 70\%.}

\item[(iii)] \emph{Scenario \(A\) is at least twice as likely as scenario \(B\).}

\item[(iv)]\emph{The assessment predicts a heart disease risk of $9\%$ or $17\%$, depending on a risk factor.\footnote{These estimates came from the Mayo Clinic Heart Disease Risk Calculator.}}

\item[(v)] \emph{An algorithm recommends an action from some menu: ``buying option $a$ is best for you.''}\footnote{This particular case is studied in \cite{ke2024learning}, and they show that such recommendations correspond to information constraints that are closed and convex under expected utility. We discuss their paper carefully in \autoref{literature}.}
\item[(vi)] \emph{Existing data cannot distinguish among models $\pi_1$, $\pi_2$, and $\pi_3$.}
\end{itemize}

For each of these statements, there is a natural, payoff-relevant state space about which information is clearly being conveyed. Moreover, each statement reflects the way people typically communicate and convey information. At the same time, it is unrealistic to assume that the DGP is known or the DM has precise beliefs about the conditional probability of receiving any of the above statements.  Because of this, the traditional notion of Bayesian updating is not defined.

For concreteness, consider a decision maker evaluating the likely effect of a treatment, such as a drug or surgical procedure, for a single patient. The natural state space is $S=\{n,p,c\}$, where $n$, $p$, and $c$ denote no, partial, or complete improvement, respectively. The DM has prior beliefs based on her medical knowledge and past experience. Now suppose an AI-assisted clinical assessment reports that, among patients with a similar profile, 80\% responded with partial or complete improvement. This information is clearly relevant. However, it does not correspond to an event in the standard state space. Further, given the inherent complexity of the setting, it is implausible to assume that the DM has a well-defined belief over the success of the treatment and the results of all possible AI-reports. Because of these features, the standard Bayesian approach is not well-defined. Nevertheless, it is reasonable for her to align her beliefs to reflect the report. Since she cannot perform Bayes' rule, another approach is needed to model belief updating.

In this paper, we introduce and behaviorally characterize \textbf{Inertial Updating}, a model of belief updating for settings with an unknown DGP in which information takes the form of natural statements, such as those above, which we refer to as \emph{general information}. More formally, \textbf{general information} is a set of probability distributions over payoff-relevant states.  A DM following Inertial Updating behaves as if she selects a revised (or posterior) belief from the general information, which acts as a constraint set. Importantly, this new belief is of minimal subjective distance from her prior, subject to the constraints imposed by the information. We establish our behavioral characterization of Inertial Updating, combined with subjective expected utility, within the \cite{anscombeaumann1963} framework.\footnote{In a companion paper, \cite{dominiak2023inertial}, we provide a characterization of Inertial Updating for standard events. \cite{dominiak2023inertial} focuses on providing a unified approach to understanding non-Bayesian updating and reactions to zero-probability events in the standard information environment. \cite{dominiak2023inertial} is not concerned with an unknown DGP and general information, which is the main theme of the current paper. The current paper focuses on updating under general information, Bayesian disagreements, and proposing and characterizing a notion of Bayesian updating within this general framework. Because the notion of information in the two papers is quite different, both conceptually and technically, the axioms and proofs are also different.}

To further illustrate our setting, it is instructive to contrast our notion with the standard Bayesian framework. In the standard framework, information is represented as an event, or a subset of some grand state space. For a set of states $S$, standard information corresponds to the statement, ``the event $E \subset S$ was realized.''\footnote{This general setting includes the standard signal structure as a special case, where each signal generates a particular event.} The typical interpretation is that the DM has learned that the true state is an element of $E$ and that states outside of $E$ are no longer possible. Returning to our medical treatment example, the natural state space is $S=\{n,p,c\}$. If a Bayesian DM were told that ``the patient is responding to the treatment,'' then this corresponds to the event $E=\{p,c\}$ and to the information constraint $I^1_E=\{\pi\in\Delta(S)\mid \pi(\{p,c\})=1\}.$ In this case, the DM has learned that state $n$ is no longer possible, and a Bayesian posterior from prior $\mu$ is
\[
\mu_{I^1_E}
=
\left(
0,
\frac{\mu(p)}{\mu(p)+\mu(c)},
\frac{\mu(c)}{\mu(p)+\mu(c)}
\right).
\]

However, in many real-life circumstances, DMs receive information in more general or nuanced forms. For example, an AI-clinical assistant may report that 80\% of similar patients responded positively to treatment. This statement is not an event in the standard state space, but it has a natural representation as a general information constraint. 
$I^{0.8}_E=\{\pi\in\Delta(S)\mid \pi(\{p,c\})=0.8\}.$
Geometrically, $I^{0.8}_E$ is a parallel shift of the event information constraint $I^1_E$, and we call such information constraints $\beta$-events (see \autoref{iuseu_fig1} for an illustration). This, and each of the general information examples discussed previously, can be represented by a subset of $\Delta(S)$.

\begin{figure}[h]
	\centering
    \includegraphics[width=6.5cm]{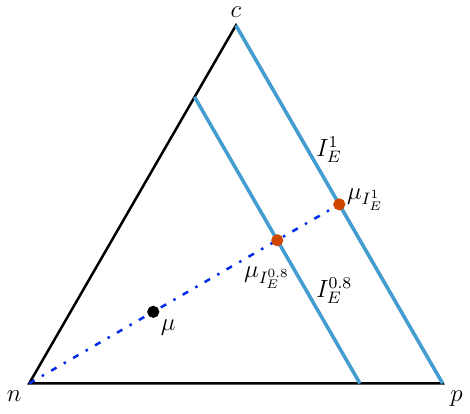}
          \includegraphics[width=6.5cm]{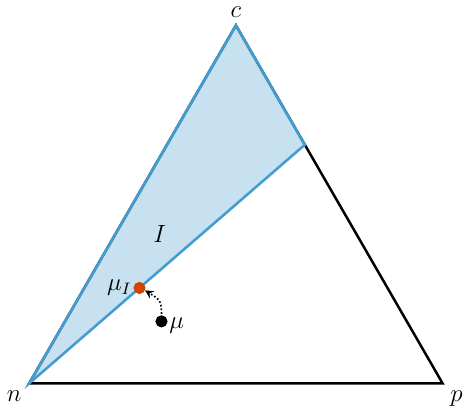}
                   \caption{Information constraints and Inertial Updating. The left panel illustrates two $\beta$-events, $I^{0.8}_E$ and $I^1_E$. The right panel illustrates a qualitative information constraint: ``$c$ is at least twice as likely as $p$,'' or $I=\{\pi\in\Delta(S)\mid\pi(\{c\}) \geq 2 \,\pi(\{p\}) \}$.}
                   \label{iuseu_fig1}
          \end{figure}

 To summarize, we develop a model of belief updating that separates the content of information from the mechanism generating it. The DGP is unknown and our DM does not have a joint prior over payoff-relevant states and information, hence Bayesian updating is not defined. This is because the knowledge of the DGP is (a) implausible and (b) ill-suited for general information.   
 
\begin{itemize}
\item[(a)] Assuming a joint prior requires the DM to have perfect knowledge of what types of information could be communicated and precisely how the distribution of statements depends upon the payoff-relevant states. This assumption is, at best, unrealistic, especially in environments which are complex (e.g., health care), the DM lacks expertise,  or she does not understand the mechanism through which information is generated (e.g., algorithmic or AI forecasts). Consequently, we should regard conclusions built upon these foundations with some caution.
\footnote{A similar point was also made by \citet{ComptePostlewaite2012}, who argue that knowledge of the DGP in individual decision making is as demanding as common knowledge in games. Indeed, criticisms of common knowledge resulted in an extensive literature on robust mechanism design and robust equilibrium analysis (e.g., see \cite{bergemann2005robust}). A complementary argument is provided by \citet{liu2026complexitybayesiansignalprocessing}, who shows that even when the signal-generating mechanism is understood, extracting the likelihood information required for Bayesian behavior is computationally intractable.}

\item[(b)] Imposing the joint-prior assumption with general information amounts to assuming that the DM has a joint belief over payoff-relevant states and possible messages, i.e., a distribution in $\Delta(S\times \mathscr I)$, where each message $I\in\mathscr I$ is a general information constraint. However, this environment is too permissive. We establish an ``anything goes'' result in Section \ref{sec:anything}: essentially any belief updating rule, and consequently any behavior, is consistent with Bayesian updating when general information is interpreted as a message.\footnote{Our result is similar conceptually to the ``anything goes'' result of \citet{ShmayaYariv2016}, who show that any behavior is rationalizable without further assumptions on how lab subjects frame an experiment. See \cite{molavi2021empirical} also for a more general result.} Hence, one needs to discipline the connection between information constraints and information structures, which is essentially what we do indirectly in this paper. 
\end{itemize}

Our main contribution is a complete behavioral analysis of Inertial Updating, along with the introduction of a notion of Bayesian updating with general information.  The behavioral characterization is achieved via three novel axioms,  \nameref{COMP}, \nameref{CON}, and \nameref{IB}, in addition to a few standard and technical postulates.  The axioms of \nameref{COMP} and \nameref{CON} ensure that the DM accepts the information as truthful. Consequently, she selects a posterior belief that is consistent with $I$. Our other crucial axiom, \nameref{IB}, ensures that beliefs are consistent with the minimization of some subjective distance notion.

Inertial Updating assumes the DM fully trusts the information source, or treats the information as \emph{hard} information, thereby treating the information constraint as a firm constraint. This trust is the main implication of the \nameref{COMP} axiom. In other settings, it is natural for the DM to only partially adjust to the information. This may be the case when information is \emph{soft} or the information source is imperfect. We therefore weaken \nameref{COMP} to obtain weighted Inertial Updating in section \ref{wIU}. In this representation, the DM's posterior is a convex combination of her prior and the IU posterior.

Although Bayes' rule is not defined in our setting, we propose a definition of Bayesian updating suitable for this environment. Intuitively, an Inertial Updater is ``Bayesian" when the selected belief maintains a form of relative consistency with the prior. We utilize this intuition and exploit information constraints that are ``event-like'' (i.e., $\beta$-events) to define a notion of Bayesian updating and characterize the corresponding class of Bayesian distances (see \autoref{iuseu_fig1}). Returning to the medical example, intuitively, a Bayesian updater adjusts her beliefs to be consistent with $\pi(\{p,c\})=0.8$ while preserving the relative likelihood of $p$ and $c$. In other words, the posterior after general information $I^{0.8}_{E}$ is
\[\mu_{I^{0.8}_E}=\left(0.2, \frac{\mu(p)}{\mu(p)+\mu(c)}\,0.8, \frac{\mu(c)}{\mu(p)+\mu(c)}\,0.8\right).\]

We use this intuition to define our notion of Bayesian updating, which we refer to as \textbf{Extended Bayesian Updating}. It reduces to Bayes' rule for events and is characterized by \nameref{IDC}, an axiom that extends the classic Dynamic Consistency to all $\beta$-events. We introduce an intuitive and tractable collection of Bayesian distances that we refer to as \nameref{GBex}, and show that under standard technical conditions \nameref{IDC} characterizes the family of \nameref{GBex}s. 

\nameref{GBex} is often called $f$-divergence and is one of the most widely used families of divergences in information theory (\cite{csiszar1967information}). Our characterization provides a behavioral foundation for the family of $f$-divergences and essentially uniquely pins down this popular class of divergences through \nameref{IDC}, which extends Bayes' rule to $\beta$-events. Notably, $f$-divergences include the widely applied Kullback-Leibler divergence, Jensen-Shannon divergence, and $\alpha$-divergence that is behaviorally equivalent to the well-known R\'enyi divergence in our setting. To the best of our knowledge, \textbf{we provide the first behavioral characterization of $f$-divergence} in economics.
 
An important implication of our characterization of \nameref{GBex} is that Bayesian updating for $\beta$-events is not sufficient to pin down a unique distance; therefore, there is no unique Bayesian distance. Consequently, richer forms of information and additional behavioral restrictions are needed to distinguish between distances. We characterize Kullback-Leibler divergence in \autoref{bayesian} via an independence condition imposed upon qualitative information constraints. Our characterization provides a complementary behavioral foundation for the prominent role of KL divergence in statistical estimation and model selection.\footnote{KL divergence plays an important role in maximum likelihood estimation, model selection criteria such as AIC, and Bayesian learning under model misspecification  \citep{white1982,akaike1974,berk1966}. In these settings, KL divergence arises from log-likelihood maximization, predictive accuracy, or asymptotic learning, respectively. In our setting, KL minimization emerges from consistency conditions on belief revision under general information.} Further, by generalizing a notion of proportionality, we provide a behavioral characterization of $\alpha$-divergence (behaviorally equivalent to R\'enyi divergence). To the best of our knowledge, \textbf{our result is the first behavioral foundation of $\alpha$-divergence} in economics. 

In Section \ref{signal}, we show that our framework nests the standard Bayesian signal-structure framework with a known DGP. When the DGP is known, each message can be represented as a linear information constraint such that Inertial Updating with a Bayesian divergence yields the Bayesian posterior. 

Another implication of our characterization is that there may be Bayesian disagreement for certain types of information constraints. Indeed, while it can be shown that  \nameref{GBex}s agree for interval information (e.g., statement (ii)), other information constraints, including qualitative information (e.g., statement (iii)), or non-convex information constraints (e.g., statement (iv)), may cause disagreement. That is, Bayesian DMs may start with the same prior, receive the same information, and arrive at different posteriors. Hence, \textbf{we identify a new source for Bayesian disagreement: general information and non-uniqueness of Bayesian divergence}. This result shows that the structure of information is crucial for the generation of polarization and disagreement.\footnote{Somewhat similarly, \cite{Baliga2013} show that polarization may arise under ambiguity (e.g., imprecise beliefs). We show that polarization may arise from imprecise information, and thus view our results as complementary.}

The remainder of this paper is structured as follows. We briefly discuss experimental evidence related to our setting in \autoref{expev}, and then in \autoref{model} we formally introduce the framework and Inertial Updating, examples of distance functions, and our ``anything goes'' result. We provide behavioral foundations of the \nameref{iuseu} representation in \autoref{axioms}. We discuss Bayesian updating in \autoref{bayesian}, including the characterizations of Kullback-Leibler and $\alpha$-divergences. In \autoref{signal}, we show how to nest the standard Bayesian framework with a signal structure in our setting. Then we discuss disagreement and speculative trade in \autoref{agreement}. We close with a discussion of related work in \autoref{literature}. The proofs are in \autoref{proofs:main}.

\subsection{Experimental Evidence}\label{expev}

Most experimental studies of belief updating provide subjects with a known data-generating process, and therefore do not directly study the environment considered here. The most direct evidence of which we are aware comes from \cite*{kovach2026recommendations}, who study belief revision from qualitative AI recommendations when the DGP is unknown. They find substantial inertia in prior beliefs and show that (weighted) Inertial Updating outperforms a quasi-Bayesian benchmark with a known DGP. Moreover, they find a threshold updating pattern (i.e., no updating when information confirms the prior, $\mu\in I$, and significant updating towards the information when it contradicts the prior, $\mu\not\in I$), which is consistent with (weighted) Inertial Updating. 

We assume that the DM only uses the qualitative content of the information rather than trying to infer the DGP from information. This assumption is useful to isolate ``learning from qualitative information when the DGP is unknown," which is the focus of our paper, and distinguish it from other forces such as ``learning about the DGP" and ``partial knowledge of the DGP." Regardless, we later show in Section \ref{signal} that our framework with an unknown DGP still ``nests" the standard information environment with a known DGP. Moreover, \cite*{kovach2026recommendations} provides empirical support for our assumption that ``the DM only uses the qualitative content of the information rather than trying to infer the DGP from information." Consistent with this assumption, they find that (a) subjects respond strongly to the qualitative content of the recommendations, (b) subjects with and without knowledge of the DGP behave almost identically, (c) and there are no dynamic effects --- inconsistent with learning about the DGP from information. 

Several other papers in similar environments also support key features of (weighted) Inertial Updating. \citet{agarwal2023combining} and \citet{vodrahalli2022humans} both study the effects of AI recommendations, although in both cases the AI advice is a single probability distribution (i.e., $I=\{\pi\}$).  \citet{agarwal2023combining} find that professional radiologists underweight AI predictions, consistent with weighted Inertial Updating. \citet{vodrahalli2022humans} find similar updating in response to otherwise-equivalent advice labeled as coming from either humans or AI. Studies using the Judge-Advisor System introduced by \cite*{sniezek1995cueing} find that DMs tend to underweight human advice relative to their own initial opinions \citep*{ bonaccio2006advice,sniezek2001trust, YanivKleinberger2000}, which is consistent with the idea behind weighted Inertial Updating.

\section{Model}\label{model}
\subsection{Basic Setup}

We study belief updating in the standard \cite{anscombeaumann1963} framework of choice under uncertainty. A DM faces uncertainty described by a nonempty and finite set of states of nature $S=\{s_1,\ldots,s_n\}$.  A nonempty subset $E$ of $S$ is called an event. Let $X$ be a nonempty, finite set of outcomes and $\Delta(X)$ be the set of all lotteries over $X$, $\Delta(X) := \Big\{ p : X \rightarrow [0,1] \mid \sum_{x \in X}p(x)=1 \Big\}$. For any Bernoulli utility function $u:X \rightarrow \mathds{R}$, by $u(p)$ we mean the expected utility of lottery $p$.

We model the DM's preference over acts. An act is a mapping $f: S \to \Delta(X)$ that assigns a lottery to each state. The set of all acts is $F:=\{f: S\to \Delta(X)\}$. Using a standard abuse of notation, any lottery $p$ also corresponds to the constant act that returns $p$ for all states.  A preference relation over $F$ is denoted by $\succsim$. As usual, $\succ$ and $\sim$ are the asymmetric and symmetric parts of $\succsim$, respectively. For any $f, g\in F$ and $E\subset S$, $f\, E\, g$ is the act that returns $f(s)$ for states $s$ in $E$ and $g(s)$ for states $s$ in $E^c:=S\setminus E$.

We denote by $\Delta(S)$ the set of all probability distributions on $S$. For notational convenience, for each $\mu\in \Delta(S)$ and each $s_i \in S$, we will sometimes write $\mu_i$ in place of $\mu(s_i)$: the probability of state $s_i$ according to $\mu$. For any set $A$ and a function $d$ on $A$, we write  $\arg\min d(A)=\{x\in A \mid d(y)\ge d(x)\text{ for any }y\in A\}$ (whenever this is well-defined). Similarly, for any set $A$ and a preference relation $\succsim$ on $A$, let $\min(A, \succsim)\equiv \{x\in A\mid y\succsim x\text{ for any }y\in A\setminus\{x\}\}$. Finally, let $\|\cdot\|$ denote the Euclidean norm, and write $B(\pi, \epsilon):=\{\pi'\in \Delta(S)\mid ||\pi-\pi'||\le \epsilon\}$ for any $\pi\in\Delta(S)$ and $\epsilon>0$.

Throughout the paper, we assume that $\mu$ has full support. Behaviorally, this means that every nonempty event $E\subset S$ is $\succsim$-nonnull, i.e., $f\, E\, g\not\sim g$ for some $f, g\in F.$  

\subsection{Information Constraints}

We consider an environment in which the DM receives new pieces of information about the uncertainty she faces (i.e., the states in $S$). Importantly, we explicitly develop a general information structure, defined below. 

\begin{defn} We call $I\subseteq \Delta(S)$ an \textbf{information constraint} if it is non-empty,  closed, and either convex or finite. The collection of all information constraints is denoted $\mathscr{I}$.
\end{defn} 

This domain covers two distinct but natural forms of information. First, information may be in the form of events, probabilistic statements, algorithmic recommendations, or inferred from observed actions of more informed decision makers, all of which naturally generate closed convex subsets of $\Delta(S)$. Second, information may consist of a set of possible distributions, models, forecasts, or opinions. This information may be non-convex, but is naturally captured by a finite set of distributions.\footnote{By excluding infinite non-convex sets, we avoid imposing behavioral assumptions on how the DM compares disconnected, infinite components. Such components may arise through a combination of mutually exclusive qualitative statements --- for example, the probability of $E$ is either below $0.2$ or above $0.8$. Further, our online appendix develops a subjective domain which consists of convex sets and finite sets that do not contain the prior in their convex hulls.}

Our DM has an initial prior over the states, and then she receives an information constraint $I \in \mathscr{I}$. In other words, the DM learns that some probability distributions are impossible. This is analogous to the standard setup, in which the DM is informed that certain states of nature are no longer possible (i.e., the DM is informed that $E \subset S$ has occurred). However, our notion of information is more general and nests the idea of an event. The information constraint containing all probability distributions concentrated on $E$, $I = \{\pi \in \Delta(S) \mid \pi(E)=1\}$, is equivalent to learning that all states outside $E$ are impossible. We will, therefore, refer to such information constraints as \emph{event information}. Our setting allows for information constraints that capture richer statements, and we provide some examples below.

\begin{itemize}
\item[(i)] ($\beta$-Event) For any $E \subset S$ and $\beta \in [0,1]$, the information constraint ``$E$ occurs with probability $\beta$'' is 
\begin{equation}
I^\beta_E=\{\pi \in \Delta(S) \mid \pi(E)=\beta\}.
\end{equation}We will refer to such an information constraint as a \emph{$\beta$-event}.\footnote{For instance, in the classic 3-color Ellsberg experiment, subjects are often informed that a ball will be drawn from an urn containing $90$ red, blue, and green balls, and exactly $30$ are red. This may be viewed as a $\beta$-event for $E=\{R\}$ and $\beta=1/3$. The idea of $\beta$-events goes back to \cite{jeffrey1965}, who used it to describe soft or imprecise information.} When $\beta=1$, this corresponds to event information. 
\item[(ii)] (Precise information) For any probability distribution $\pi$, the information constraint ``$\pi$ is the true distribution'' is 
\begin{equation}I=\{\pi\}.\end{equation}
\item[(iii)] (Qualitative information) For any $A,B \subset S$, the information constraint ``$A$ is at least $\gamma$-as likely as $B$'' is
\begin{equation}I=\{\pi \in \Delta(S) \mid \pi(A)\ge \gamma\,\pi(B)\}.\end{equation} 
Notice that for $\gamma=1$, this corresponds to the classic notion of qualitative information. 
\item[(iv)] (Interval information) For $E \subset S$, and $0\le \underline{\beta}<\overline{\beta} \le 1$, interval information corresponds to ``the probability of $E$ is between $\underline{\beta}$ and $\overline{\beta}$'' is \begin{equation} I=\{\pi \in \Delta(S) \mid  \underline{\beta} \leq \pi(E) \leq \overline{\beta}  \}.\end{equation} 

\end{itemize}

\subsection{Belief Revision and Inertial Updating}

The DM's behavior is represented by a family $\{\succsim_I\}_{I\in\mathscr{I}}$ of preference relations, each defined over $F$.  Before any information is revealed (i.e., $I=\Delta(S)$), we write $\succsim$ in place of $\succsim_{\Delta(S)}$, and we call $\succsim$ the initial preference. As new information $I \subset \Delta(S)$ arrives, the DM revises $\succsim$ given $I$. The new preference, denoted $\succsim_I$, governs the conditional behavior of the DM in light of $I$. 
 
We assume that the DM's initial preference is of the SEU form. That is, the initial preference $\succsim$ admits a SEU representation with respect to a Bernoulli utility function $u: X \to \mathds{R}$ and a (unique) probability distribution $\mu\in \Delta(S)$ such that for any $f, g\in F$, 
\begin{equation}
f\succsim g\text{ if and only if }\sum_{s\in S}\mu(s) u\big(f(s)\big)\ge \sum_{s\in S}\mu(s) u\big(g(s)\big).\end{equation}
Hence, the DM's initial behavior is characterized by the pair $(u,\mu)$. 

How does the DM's behavior change after she learns $I$? We assume that the DM updates her initial preference $\succsim$ by revising her initial belief $\mu$ while keeping her risk attitude unchanged. Let $\mu_I$ denote the DM's revised (updated) belief conditional on $I$. How does the DM form her new belief $\mu_I$ when $\mu \not \in I$, that is, her old belief $\mu$ conflicts with the available information?

We impose two properties on the DM's belief revision. First, we assume that the DM accepts the information, so that her new belief $\mu_I$ is consistent with  $I$ (i.e., $\mu_I \in I$). Second, we assume that she exhibits ``inertia of initial beliefs." Therefore, she chooses the $\mu_I$ closest to her initial belief $\mu$. That is, the DM forms her new belief $\mu_I$ as if she was minimizing the distance between her initial belief $\mu$ and all probability distributions consistent with $I$ (see \autoref{iuseu_fig1}).\footnote{This behavior may be interpreted as an implication of the principle of insufficient reason (or principle of indifference); the DM, without a way to assign likelihoods to $\pi \in I$, decides to make the minimal change to be consistent with $I$.} We call this updating procedure \textbf{Inertial Updating}. When her prior belief $\mu$ is consistent with $I$, the closest probability measure is $\mu$, and thus the DM keeps it (i.e., $\mu=\mu_I$ when $\mu\in I$).

Combining these assumptions, our DM admits a SEU representation with respect to some Bernoulli utility function $u$ and, for each information constraint $I\in\mathscr{I}$, an updated belief $\mu_I \in I$ that is of minimal distance from her initial belief $\mu$. Additionally, we require that the DM's notion of distance satisfy two intuitive properties, which we formally define below.


\begin{defn}[Distance Function and Its Tie-Breaker] A function $d:\Delta(S) \to\mathds{R}$ is a \textbf{distance function} with respect to $\mu\in \Delta(S)$, denoted by $d_\mu$, if, \begin{itemize}
\item[(i)] for any distinct $\pi, \pi'\in\Delta(S)$ with $d_{\mu}(\pi)=d_{\mu}(\pi')$, there is some $\alpha\in (0, 1)$ such that $d_{\mu}(\pi)>d_{\mu}(\alpha \pi+(1-\alpha)\pi')$;
\item[(ii)] $d_{\mu}(\mu)<d_{\mu}(\pi)$ for any $\pi\in\Delta\setminus\{\mu\}$. 
\end{itemize} Moreover, a strict total order $\succ_{d_\mu}$ on $\Delta(S)$ is a \textbf{tie-breaker for $d_\mu$} if for any $\pi, \pi'\in \Delta$, $d_{\mu}(\pi)>d_{\mu}(\pi')$ implies $\pi\succ_{d_\mu} \pi'$. 
\end{defn}

Property (i) requires that if two beliefs are equidistant from $\mu$, then there is a mixed belief that is strictly closer to $\mu$. This mild property is natural and weaker than strict quasiconvexity. Property (ii) ensures that $\mu$ is the unique global minimizer and that all other beliefs are strictly further away. We also introduce a notion of tie-breaking. As we will show, tie-breaking is only necessary for non-convex (i.e., finite) information constraints. That is, under \nameref{iuseu} with continuous $d$, property (i) ensures that each convex $I \in \mathscr{I}$ admits a unique minimizer:  $\mu_I=\arg\min d_\mu(I)$. We now formally define our model.

\begin{defn}[{\bf{Inertial Updating}}]\label{iuseu} A family of preference relations $\{\succsim_I\}_{I\in\mathscr{I}}$ admits an \textbf{Inertial Updating} subjective expected utility representation if there are a Bernoulli utility function $u:X\to\mathds{R}$, a prior $\mu \in \Delta(S)$, a distance function $d_\mu:\Delta(S)\to\mathds{R}$, and its tie-breaker $\succ_{d_\mu}$ such that for each $I \in \mathscr I$, the preference relation $\succsim_I$ admits a SEU representation with $(u, \mu_I)$ where
\[\mu_I\equiv \min(\arg\min d_\mu(I), \succ_{d_\mu}).\]
 That is, for any $f, g\in F$, 
\[
f \succsim_I g \quad \text{if and only if} \quad \sum_{s\in S}\mu_I(s) u\big(f(s)\big)\ge \sum_{s\in S}\mu_I(s) u\big(g(s)\big).
\]

\end{defn}

Note that when $d$ is continuous, the argmin is nonempty because $I$ is compact. For convex $I$, property (i) implies that the minimizer is unique; for finite \(I\), the tie-breaker $\succ_{d_\mu}$ selects uniquely from the finite set of minimizers.

The family $\{\succsim_I\}_{I\in\mathscr{I}}$ of \nameref{iuseu} preferences is characterized by $(u, d_\mu, \succ_{d_\mu})$.  We will restrict our attention to locally nonsatiated distance functions as defined below. 

\begin{defn} A distance function $d_\mu$ is \textbf{locally nonsatiated} with respect to $\mu$ if for any $\pi\in\Delta(S)\setminus\{\mu\}$ and $\epsilon>0$, there is a $\pi'$ such that $\|\pi-\pi'\|<\epsilon$ and $d_\mu(\pi')<d_\mu(\pi)$. 
\end{defn}

\subsection{Notions of Distance}\label{notdis}

Our DM's notion of distance is subjective. Thus, our framework allows for a wide variety of distance notions and, consequently, a wide variety of updating behaviors. In this section, we discuss a few examples of distance functions.  We begin by defining a generalization of Kullback-Leibler divergence, a well-known method of measuring the distance between probability measures from information theory.   

\begin{defn}[\textbf{Bayesian divergence}]\label{GBex} For a strictly concave, differentiable function $\sigma:\mathds{R}_{+}\to\mathds{R}$ with $\sigma(1)=0$, let $d_{\mu}$ be given by
\begin{equation}
d_\mu(\pi)=-\sum^n_{i=1}\mu_i\, \sigma(\frac{\pi_i}{\mu_i}).
\end{equation}
\end{defn}

 The \nameref{GBex} is also called the $f$-divergence and is one of the most commonly used divergences in information theory (\cite{csiszar1967information}).\footnote{Technically, $f$-divergence is defined as $D_f(\mu||\pi)=\sum^n_{i=1}\pi_i\, f(\frac{\mu_i}{\pi_i})$ for a convex function $f$ such that $f(1)=0$. Note that we obtain \nameref{GBex} by simply setting $\sigma(x)=-x\,f(\frac{1}{x})$. Moreover, $f''\ge(>) 0$ iff $\sigma''\le(<) 0$, and $f(1)=0$ iff $\sigma(1)=0$. Throughout the paper, we assume that $\sigma$ is strictly concave, or equivalently, $f$ is strictly convex, which is satisfied for most of the well-known $f$-divergences.} Notice that \nameref{GBex} reduces to the KL divergence when $\sigma(x) = \ln(x)$ and $\alpha$-divergence when $\sigma(x)=\frac{(1-\alpha)+\alpha\,x-x^\alpha}{\alpha(\alpha-1)}$.\footnote{When $\alpha=0$, we obtain $\sigma(x)=\log(x)-x+1$, which gives the KL divergence, and when $\alpha=1$, we obtain $\sigma(x)=x-1-x\log(x)$, which gives the reverse KL divergence. Any distance function with $\sigma(x)=A\,\frac{x^\alpha}{\alpha(1-\alpha)}+B\,x+c$ with $A>0$ will be called $\alpha$-divergence because $\sigma(x)$ and $\tilde{\sigma}(x)=A\sigma(x)+B\,x+C$ lead to the same distances up to an affine transformation, i.e.,  $d^{\tilde{\sigma}}_\mu(\pi)=A\,d^{\sigma}_\mu(\pi)-B-C$.} In our setting, any R\'enyi divergence is behaviorally equivalent to an $\alpha$-divergence with some $\alpha<1$ (and $\alpha\neq 0$). When $\alpha=0$, $\alpha$-divergence reduces to the KL divergence.\footnote{As is well known, the KL divergence and $\alpha$-divergence with $\alpha<0$ are not well defined on the boundary of $\Delta(S)$ since $\sigma(0)=-\infty$. Hence, technically, we need to transform it into a proper distance function. For example, $d_\mu(\pi)=\Phi(-\sum_{s_i\in \text{supp}(\pi)}\mu_i\, \sigma(\frac{\pi_i}{\mu_i}))-|\text{supp}(\pi)|$, where $\Phi$ is the sigmoid function, is a distance function that is behaviorally equivalent to $d^\sigma_\mu$ on $\text{int}\Delta(S)$. Under this formulation, our subsequent characterization results on Bayesian divergences hold.} This distance function is particularly easy to interpret when information is in the form of a $\beta$-event: $I^\beta_E=\left\{\pi\in \Delta(S) \mid \pi(E)=\beta\right\}$. 

\begin{prps}\label{prop:bayes} For any \nameref{iuseu} representation with \nameref{GBex} and for any $E\subset S$ and $\beta\in (0, 1)$,
\begin{equation}
\mu_{I^\beta_E}(s)=\beta\,\frac{\mu(s)}{\mu(E)}\mathds{1}\{s\in E\}+(1-\beta)\,\frac{\mu(s)}{\mu(E^c)}\mathds{1}\{s\in E^c\}.
\end{equation}

\end{prps}

That is, the DM shifts probability mass between events $E$ and $E^c$, maintaining the relative probabilities between states within $E$ (and $E^c$).\footnote{This form of belief revision for probabilistic information is also known as Jeffrey Conditionalization \citep{jeffrey1965}.}  For $\beta=1$ or $0$ (i.e., $I^\beta_E$ is standard event information), the Bayesian divergence\footnote{We require $\sigma(0)$ to be finite. If not, we can utilize the sigmoid transformation discussed in footnote 14 to extend to the boundary of $\Delta(S)$ and the Bayesian posterior is correctly selected.} yields Bayesian updating on $E$ and $E^c$, respectively. However, for more general information constraints, the precise form of $\sigma$ will matter. Consequently, ``Bayesian'' DMs might disagree with each other when provided more general information constraints $I$. Bayesian divergences are studied in \autoref{bayesian}, while disagreement is discussed in \autoref{agreement}.

An example of an admissible distance that is distinct from the \nameref{GBex} is the Euclidean distance. 

\begin{defn}[\textbf{Euclidean}]\label{euclidian} Let $d_\mu(\pi)=\sum^n_{i=1} (\pi_i-\mu_i)^2$.\end{defn}

For the Euclidean distance, when the minimization problem has an interior solution, for each $s_i \in E$, \[\mu_{I_E^\beta}(s_i)=\mu_i+\frac{\beta-\mu(E)}{|E|}.\]  Here, prior odds are ignored when updating beliefs and the net-probability change,  $\beta-\mu(E)$, is allocated to the states uniformly. For $s \in E^c$, the residual probability is subtracted uniformly. This distance is a special case of Bregman divergence.

At an interior solution, Euclidean updating allocates the net probability change uniformly over the relevant states, rather than proportionally to the prior, hence it is non-Bayesian. By modifying the distance in particular ways, Inertial Updating may capture a variety of non-Bayesian belief patterns. The connection to non-Bayesian behavior is explored in our companion paper \citep{dominiak2023inertial}, while this paper focuses on general information and Bayesian divergences.

\subsection{Interpreting General Information as a Message}\label{sec:anything}

In principle, it is possible to interpret general information as a signal or message and then apply Bayes' rule in this richer setting. While this approach seems reasonable at first inspection, as we will show below, it does not provide any guidance on belief updating rules for general information when the DGP is unknown. This is because essentially, any updating rule can be written as a result of Bayes' rule by appropriately choosing the joint distribution over signals and states. Our result is similar to that of \citet{ShmayaYariv2016}, who show that Bayes' rule does not have any testable predictions without imposing assumptions on the information structure. 

To show this, consider the following mapping $\rho_\mu:\mathcal{I}\to\Delta(S)$ where $\mathcal{I}$ is an arbitrary finite subset of $\mathscr{I}$ that contains all event information constraints. Here, $\rho_\mu$ is an updating rule where $\rho_\mu(I)$ is the posterior distribution after receiving information $I$. We say an updating rule $\rho_\mu$ has a \emph{Bayesian representation} if there is an information structure $P:S\to \Delta(\mathcal{I})$ such that 
\[\rho_\mu(I)(s)=\frac{P(I|s)\mu(s)}{\sum_{s'\in S}P(I|s')\mu(s')}.\]

To demonstrate that the Bayesian representation itself will not impose much structure on updating rule, let us start with the following example. 

\begin{exm} Consider $\rho_\mu$ such that for some $\epsilon>0$,  $\rho_\mu(I^1_E)(s)=\mu(s)+\epsilon(E, s)$ for any $s\in E$ and $\rho_\mu(I^1_E)(s)=\mu(s)-\epsilon(E, s)$ for any $s\in E^c$, for every $E\subset S$. Then, $\rho_\mu$ has a Bayesian representation.
\end{exm} 

The preceding example illustrates that as long as the updating rule responds to the qualitative information provided by ``standard information,'' the updating rule admits a Bayesian representation.

Since $I$ contains qualitative information about the likelihoods of states, there must be a certain connection between information contained in $I$ and $\rho_\mu(I)(s)$. For example, the classic axiom of Consequentialism requires that $\rho_\mu(I_E)(s)=0$ if $I_E=\{\pi\in \Delta(S) \mid \pi(E)=1\}$ and $s\in E^c$. We impose the following minimal requirements. 

\begin{defn}[\textbf{Minimal Responsiveness}]\label{MR} For any non-empty $E$, $\rho_\mu(I_E)(s)\ge \mu(s)$ for all $s\in E$, and $\rho_\mu(I_E)(s)<\mu(s)$ for all $s\in E^c$.   
\end{defn}

\begin{prps}\label{prop:anything} Any minimally responsive $\rho_\mu$ has a Bayesian representation. 
\end{prps}

The idea behind the result is straightforward. We first show that a Bayesian representation exists if and only if the prior can be expressed as a convex combination of the posteriors $\rho_\mu(I)$ with strictly positive weights. These weights give the unconditional probabilities of the corresponding messages. Using a separating hyperplane argument, we can show that  \nameref{MR} ensures that the prior is interior to the convex hull of posteriors, hence strictly positive weights exist. While it is clear from the proof of Proposition \autoref{prop:anything} that \nameref{MR} can be relaxed or modified, finding weaker conditions under which Proposition \autoref{prop:anything} holds is somewhat orthogonal to the issue of finding a Bayesian representation without placing additional properties on the information structure.

The proposition shows that the standard approach, Bayes' rule, imposes essentially no structure on updating rules for general information when the DGP is unknown. In contrast, we impose structure and discipline on belief updating rules for general information using \nameref{iuseu}.

\section{Axiomatic Characterization}\label{axioms}

In this section, we present behavioral postulates that characterize the family of Inertial Updating SEU preferences. Our first axiom imposes the standard SEU conditions of \cite{anscombeaumann1963} on each (conditional) preference relation $\succsim_I$ in addition to \textbf{Invariant Risk Preferences}.  Because these conditions are well understood, we will not provide a formal discussion of them. 

\begin{axm}[\bf{Standard Postulates}]\label{SEU} For each $I\in \mathscr{I}$, the following conditions hold.

\begin{itemize}
\item[$(i)$] \textbf{Weak Order:} $\succsim_I$ is complete and transitive. 
\item[$(ii)$] \textbf{Archimedean:} For any $f, g, h\in F$, if $f\succ_I g$ and $g\succ_I h$, then there are $\alpha, \beta\in (0, 1)$ such that $\alpha f+(1-\alpha) h\succ_I g$ and $g\succ_I \beta f+(1-\beta) h$.
\item[$(iii)$] \textbf{Monotonicity:} For any $f, g\in F$, if $f(s)\succsim_I g(s)$ for each $s\in S$, then $f\succsim_I g$. Further, if $f(s)\succsim g(s)$ for each $s\in S$ and $f(s)\succ g(s)$ for some $s$, then $f\succ g$.
\item[$(iv)$] \textbf{Nontriviality:} There are $f, g\in F$ such that $f\succ_I g$. 
\item[$(v)$] \textbf{Independence:} For any $f, g, h\in F$ and $\alpha\in (0, 1]$, $f\succsim_I g$ if and only if $\alpha f+(1-\alpha) h\succsim_I \alpha g+(1-\alpha) h$. 
\item[$(vi)$] \textbf{Invariant Risk Preferences:} For any $p,q\in\Delta(X)$, $p\succsim_I q$ if and only if $p\succsim q$.
\end{itemize}
\end{axm}  

Note that \textbf{Monotonicity}, which is slightly stronger than the standard monotonicity condition, implies that the prior $\mu$ has full support. \textbf{Invariant Risk Preferences} requires that the DM's preference over lotteries does not change when information is provided.  

The next two axioms ensure that the DM forms a new belief that is consistent with the available information.

\begin{axm}[\bf{Compliance}]\label{COMP}
For any $p, q\in \Delta(X)$, $E\subset S$, and $\pi\in\Delta(S)$, 
\[p\, E\, q\sim_{\{\pi\}} \pi(E)\,p+\big(1-\pi(E)\big)\,q.\]
\end{axm}

\nameref{COMP} requires that the DM adheres to precise information when provided.  That is, whenever a single distribution is provided,  $I=\{\pi\}$ for some $\pi \in \Delta(S)$, the DM adopts $\pi$.  One way to think of this axiom is to imagine a patient visiting her doctor and discussing a potential treatment. If the doctor provides extremely precise information about the treatment and its chances of success or failure, the patient accepts this information completely and adopts the doctor's information as her beliefs about the states. \nameref{COMP} resembles Consequentialism (see \cite{Ghirardato2002}) in dynamic settings under uncertainty.\footnote{The decision criterion used in \cite{lanzani2025dynamic} to analyze dynamic concern for model misspecification ``satisfies" \nameref{COMP}. One of his key axioms, Uniform Misspecification Concern, which characterizes his decision criterion, is conceptually similar to \nameref{COMP}.}

To formally state the next axiom, we first define a notion of equivalent information constraints. Given two sets of information $I$ and $I'$, we say that they are \emph{preference equivalent} if $\succsim_I=\succsim_{I'}$  (that is, $f \succsim_I g$ if and only if $f \succsim_{I'} g$ for all $f, g \in F$). In this case, we may also say that $\succsim_I$ and $\succsim_{I'}$ are equivalent.

\begin{axm}[\bf{Internal Consistency}]\label{CON} For any $I\in\mathscr{I}$, there is $\pi\in I$  such that $\succsim_{I}=\succsim_{\{\pi\}}$.
\end{axm}

\nameref{CON} requires that for each information constraint, there is a more precise information constraint that has the same informational content for the DM.\footnote{Under \nameref{IB} in addition to a mild technical condition, \nameref{CON} is equivalent to the following condition: for any $I, I'\in\mathscr{I}$ with $I\cup I'\in \mathscr{I}$, if $\succsim_I\neq \succsim_{I\cup I'}$, then $\succsim_{I'}=\succsim_{I\cup I'}$ .}

Under \nameref{SEU}, \nameref{COMP}, and \nameref{CON}, there is $(u, (\mu_I)_{I\in\mathscr{I}})$ such that for each $I$, $\mu_I \in I$ and the conditional preference $\succsim_I$ admits the SEU representation with $(u, \mu_I)$. Notice that the DM may form a new belief $\mu_I$ that is completely independent of the initial belief $\mu$.

The next axiom, \nameref{IB}, is the most important behavioral condition for our model and connects initial and new beliefs. Loosely, \nameref{IB} implies a form of ``belief consistency'' across various information constraints.  For an intuition behind \nameref{IB}, consider $E \subseteq S$ and let $I_1=\{\pi \in \Delta(S) \mid \pi(E) \ge \frac12\}$ be the information ``$E$ is more likely than $E^c$.''  This may alter the DM's preferences regarding bets on $E$. Suppose that the more refined information, $I_2=\{\pi \in \Delta(S) \mid \frac34 \ge \pi(E) \ge \frac12\}$, induces the same conditional preferences regarding bets on $E$. This suggests that the DM's willingness to bet on $E$ is not dependent on the upper bound of $I_2$  (e.g., because it is determined by the lower bound placed on the probability of $E$), and so any information constraint of the form $I'=\{\pi \in \Delta(S) \mid \beta \ge \pi(E) \ge \frac12\}$ should yield exactly the same willingness to bet on $E$ as $I_1$ for any $\beta \in [\frac34,1]$. \nameref{IB} extends this idea to more general information constraints.

	\begin{figure}[h]
	\centering
          \includegraphics[width=6.5cm]{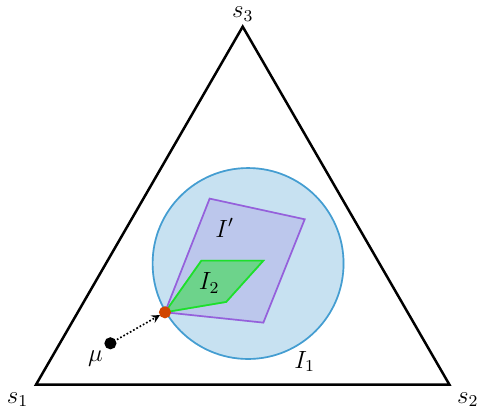}
          \includegraphics[width=6.5cm]{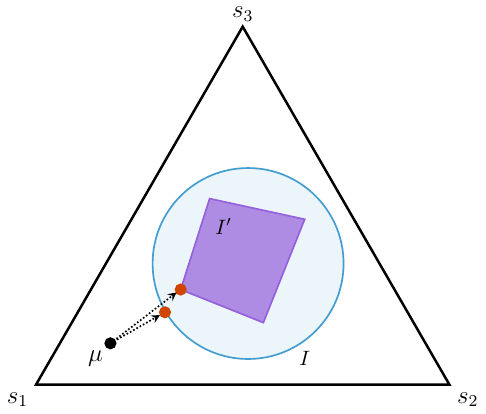}
                   \caption{Illustrations of Informational Betweenness (left) and Extremeness (right).}\label{SARBfig}
          \end{figure}

\begin{axm}[\bf{Informational Betweenness}]\label{IB} For all $I_1, I_2, I'\in\mathscr{I}$ such that $I_2\subseteq I'\subseteq I_1$,
\[\text{ if }\succsim_{I_1}=\succsim_{I_2}\text{, then }\succsim_{I_1}= \succsim_{I'}.\]
\end{axm}

In other words, \nameref{IB} requires that if the least precise and most precise information constraints, $I_1$ and $I_2$, provide ``the same information'' (i.e., are behaviorally equivalent), then any intermediate information constraint, $I'$, must also provide the same information as these two sets. This logic is illustrated in \autoref{SARBfig}.

At this point, it is worth remarking that \nameref{IB} captures most of the behavioral content of \nameref{iuseu}. That is, consider the family of SEU preferences $\{\succsim_I\}_{I\in\mathscr{I}}$ characterized by $\big(u,(\mu_I )_{I\in\mathscr{I}}\big)$. Under \nameref{IB}, it turns out that $\mu_I$ is the minimizer of some (complete and transitive) ordering, which, of course, depends on the prior $\mu$. To ensure this order is consistent with a distance function, we require three technical conditions, \nameref{EXTM}, \nameref{RCON}, and \nameref{RCONT}. 

For an intuition behind \nameref{EXTM}, imagine that the DM exhibits a change in behavior after learning $I$; she finds $I$ to be ``informative'' and changes her beliefs ($\succsim_I\neq \succsim$). Then, any ``more informative'' $I'$ must similarly result in additional changes in behavior ($\succsim_{I'}\neq \succsim_{I}$). In other words, a DM who adjusts her beliefs after some information $I$ must continue to move her beliefs when you present her with any strictly more informative information constraint $I'$.  As we have seen in our discussion of \nameref{IB}, being a (strict) subset is not sufficient for the DM to perceive a set of information as more informative because the information constraints may intersect at the chosen belief. Drawing on this insight, we suggest that interiority is the correct notion of ``more informative."  This logic is illustrated in \autoref{SARBfig}. The interior of each $I$ is denoted by $\text{int}(I)$, i.e., $\text{int}(I)=\{\pi\in I \mid \exists\epsilon>0\text{ such that }B(\pi, \epsilon)\subseteq I\}$.

\begin{axm}[\bf{Extremeness}]\label{EXTM} For any convex $I, I'\in\mathscr{I}$ with $I'\subseteq \text{int}(I)$,
\[\text{if }\succsim_I\neq \succsim\text{, then }\succsim_{I'}\neq \succsim_{I}.\]
\end{axm}

To define the next two axioms, consider a scenario where the DM receives two competing reports $\pi$ and $\pi'$. One can say that $\pi$ is robustly closer than $\pi'$ if the DM adopts $\pi$ over $\pi'$ and would adopt any small perturbation of $\pi$ over any small perturbation of $\pi'$. We then define a revealed weakly farther than relation $R$ by negating the above robustly closer than notion: $\pi R \pi'$ if $\pi$ is not robustly closer than  $\pi'$. Formally, we write $\pi\, R\, \pi'$ if there is no $\epsilon>0$ such that $\succsim_{\{\pi_\epsilon, \pi'_{\epsilon}\}}=\succsim_{\{\pi_{\epsilon}\}}$ for any $\pi_\epsilon\in B(\pi, \epsilon)$ and $\pi'_\epsilon\in B(\pi', \epsilon)$. We require that $R$ is transitive. 

\begin{axm}[\bf{Revealed Transitivity}]\label{RCONT} $R$ is transitive.
\end{axm}

The last axiom requires the revealed weakly farther than relation to be consistent with the DM's behavior under convex information constraints.

\begin{axm}[\bf{Revealed Consistency}]\label{RCON} For any distinct $\pi, \pi'\in\Delta$, if $\succsim_{[\pi, \pi']}=\succsim_{\{\pi\}}$, then $\neg(\pi\,R \pi')$. 
\end{axm}

\begin{thm}\label{repthrm} A family of preference relations $\{\succsim_I\}_{I\in\mathscr{I}}$ satisfies Axioms 1 through 7 if and only if it admits an {\bf{Inertial Updating}} representation with respect to some locally nonsatiated, continuous distance function $d_\mu$.
\end{thm}

By the uniqueness of subjective expected utility representations, $\mu$ and $\mu_I$ are unique while $u$ is unique up to a positive affine transformation. It turns out that the tie-breaker is also unique and the distance function is unique up to a strictly monotonic transformation. 

\begin{prps}\label{uniqueness} Suppose the family of preference relations $\{\succsim_I\}_{I\in\mathscr{I}}$ admits \nameref{iuseu} representations with $(u, d_\mu, \succ_{d_\mu})$ and $(u', d'_{\mu'}, \succ_{d'_{\mu'}})$ where distance functions are continuous. Then $(i)$ $\mu=\mu'$ and $\mu_I=\mu'_{I}$ for each $I\in\mathscr{I}$, $(ii)$ $u=\alpha u'+\beta$ for some $\alpha, \beta$ with $\alpha>0$, $(iii)$ $\succ_{d_\mu}=\succ_{d'_{\mu'}}$, and (iv) there is a strictly increasing function $F$ such that $d_\mu=F\circ d'_\mu$.
\end{prps}

\subsection{Weighted  Inertial  Updating}\label{wIU}

Our baseline characterization treats the information constraint as a hard constraint. This is appropriate when the decision maker regards the information as objective or sufficiently authoritative that any posterior outside the information constraint is inadmissible. In many environments, however, the source is informative but not fully trusted. This is especially natural when information is generated by an opaque recommendation system, expert advice of uncertain quality, or a mechanism whose data-generating process is not well understood. In such cases, \nameref{COMP} is too strong: the decision maker may adjust toward the information constraint without satisfying it fully because she believes the information to be \emph{soft}. We therefore introduce weighted Inertial Updating, which captures this intermediate case in a tractable way.

\begin{defn}[{\bf{Weighted Inertial Updating}}]\label{wiuseu} A family of preference relations $\{\succsim_I\}_{I\in\mathscr{I}}$ admits a \textbf{weighted Inertial Updating} (wIU) representation if there are a Bernoulli utility function $u:X\to\mathds{R}$, a prior $\mu \in \Delta(S)$, a distance function $d_\mu:\Delta(S)\to\mathds{R}$, its tie-breaker $\succ_{d_\mu}$, and $w\in [0, 1]$ such that for each $I \in \mathscr I$, the preference relation $\succsim_I$ admits a SEU representation with $(u, \mu_I)$ where
\[\mu_I\equiv w\,\mu+(1-w)\,\min(\arg\min d_\mu(I), \succ_{d_\mu}).\]
 That is, for any $f, g\in F$, 
\[
f \succsim_I g \quad \text{if and only if} \quad \sum_{s\in S}\mu_I(s) u\big(f(s)\big)\ge \sum_{s\in S}\mu_I(s) u\big(g(s)\big).
\]
\end{defn}

The weighted representation is behaviorally useful for two reasons. First, it preserves the central idea of minimal-distance belief revision while relaxing the most demanding feature of the baseline model. Second, it provides a natural bridge between our theory and recent evidence on belief updating under unknown data-generating processes. In recent experimental work on qualitative AI recommendations, \cite{kovach2026recommendations} show that weighted Inertial Updating explains the observed threshold-type updating patterns better than a quasi-Bayesian benchmark, a three-parameter generalization of Bayes' rule with a known DGP. For example, the completeness measure of \cite{fudenberg2021complete} of the wIU is 0.914, which means that 91.4\% of predictable variation (i.e., excluding unpredictable noise) of the data can be explained by the model. In contrast, the completeness measure of the quasi-Bayesian model is 0.42. Weighted Inertial Updating also predicts no updating when information confirms (i.e., contains) the DM's prior, which is also supported by the above experimental evidence. This suggests that the weighted model is an empirically relevant specification for recommendation environments and other settings in which the source is only partially trusted.\footnote{\cite{agarwal2023combining} and \cite{vodrahalli2022humans} run experiments where the subjects receive precise information $I=\{\pi\}$. Under wIU, the posterior $\mu_I=w\,\mu+(1-w)\,\pi$ is linear in the prior, which is largely consistent with their experimental data, as explained in Section 6 of \cite*{kovach2026recommendations}.}

In the wIU representation, the DM places a constant weight $w$ on her prior, which may be interpreted as a measure of conservatism or, equivalently, $1-w$ is a measure of trust in the information source. The more the DM trusts the source, the less she clings to her prior. 

To characterize this, we weaken \nameref{COMP}. Each of the following conditions is implied by \nameref{COMP} under SEU.

\begin{axm}[\textbf{Linearity}]\label{Lin} For any $p, q, r, r'\in\Delta(X)$, $E\subseteq S$, $\pi, \pi'\in\Delta(S)$, and $\lambda\in[0, 1]$,
\[\text{if }p\,E\,q\sim_{\{\pi\}}\,r\text{ and }p\,E\,q\sim_{\{\pi'\}}\,r', \text{ then }p\,E\,q\sim_{\{\lambda\,\pi+(1-\lambda)\,\pi'\}}\,\lambda\, r+(1-\lambda)\,r'\]
\end{axm}

\nameref{Lin} says that the information impact of precise information $\pi$ is mixture-independent.

\begin{axm}[\textbf{Event Independence}]\label{EInd} For any $p,q,r\in\Delta(X)$ with $p\succsim q$, $E\subseteq S$, and $\pi,\pi'\in\Delta(S)$, 
\[\text{if }\pi(E)\geq\pi'(E) \text{, then } p\,E\,q\succsim_{\{\pi'\}} r \text{ implies } p\,E\,q\succsim_{\{\pi\}} r.\]
\end{axm}

\nameref{EInd} says that the information impact of precise information $\pi$ on event $E$ only depends on $\pi(E)$.

\begin{axm}[\textbf{Initial Compliance}]\label{IComp} For any $\pi, \pi'\in \Delta(S)$ with $\succsim\neq\succsim_{\{\pi'\}}$, $p, q\in\Delta(X)$, and $E\subset S$, \[\text{if }\succsim=\succsim_{\{\pi\}}, \text{ then } p\, E\, q\sim \pi(E)p+(1-\pi(E))\, q.\]\end{axm}

The last axiom imposes \nameref{COMP} only when the precise information coincides with the DM's prior.

\begin{thm}\label{repthrm2} A family of preference relations $\{\succsim_I\}_{I\in\mathscr{I}}$ satisfies Axiom 1 and Axioms 3 through 10 if and only if it admits a {\bf{weighted Inertial Updating}} representation with respect to some locally nonsatiated, continuous distance function $d_\mu$.
\end{thm} 

When $w<1$, it is straightforward to obtain a uniqueness result that generalizes Proposition 3, where $w$ is also unique.

\section{Bayesian Updating with General Information}\label{bayesian}

As discussed in the introduction, Bayes' rule is not defined in our setting with general information constraints. Nevertheless, for information constraints that are event-like (e.g., $\beta$-events), there is a natural analog to Bayesian updating on standard events.  In section \ref{notdis}, we defined \nameref{GBex}, a distance notion that leads to posteriors consistent with Bayesian updating on standard events. In this section, we show that these are essentially the only distance notions that lead to a version of Bayesian updating defined for general information. Therefore, we interpret Inertial Updating with Bayesian divergences as the exact analog of Bayesian updating in a general information setting.

Recall that $\beta$-events are information constraints of the form $I^\beta_E\equiv \{\pi\in\Delta(S) \mid \pi(E)=\beta\}$ for some $E \subset S$. In the standard setting, revealing that an event $E$ occurred (i.e., $\beta=1$) induces a Bayesian DM to revise her prior $\mu$ by proportionally allocating all probability mass among states in $E$; i.e., $\pi_i=\frac{\mu_i}{\mu(E)}$. This principle of proportionality ought to apply to more general information constraints, including all $\beta$-events. That is, a Bayesian DM should allocate the given probability $\beta$ among states in $E$ in such a way that preserves relative probabilities. Thus, updated beliefs will still be proportional to her prior, $\pi_i=\frac{\mu_i}{\mu(E)} \beta$.  We call this updating procedure \nameref{GBU}, which we formally define below. 

\begin{defn}[\textbf{Extended Bayesian Updating}]\label{GBU} Beliefs satisfy Extended Bayesian Updating if for any $E\subset S$, $s\in E$, and $\beta\in (0, 1]$,
\[\mu_{I^\beta_E}(s)=\frac{\mu(s)}{\mu(E)}\,\beta.\]
\end{defn}

Extended Bayesian updating is illustrated in the left panel of \autoref{iuseu_fig1}. The dashed line can be thought of as the ``Bayesian expansion path'' of beliefs $\mu$. That is, it ensures that for any $\beta$, posterior beliefs after the $I^{\beta}_{E}$ information constraint must lie on this dashed line.

For the standard notion of information as an event, it is well-known that Dynamic Consistency characterizes Bayesian updating (see \cite{epstein1993} and  \cite{Ghirardato2002}). However, as our setting allows for more general information, we need to extend Dynamic Consistency beyond standard events. The key axiom, called \nameref{IDC}, requires that if a DM prefers act $f$ to act $g$ before she receives any information and the two acts coincide outside of $E$ (i.e., $f(s)=g(s)$ for all $s\in E^c$), then she prefers $f$ to $g$ after the information constraint $I_E^{\beta}$ is revealed, and vice versa. Obviously, it corresponds to the standard notion of dynamic consistency when $\beta=1$.

\begin{axm}[\bf{Informational Dynamic Consistency}]\label{IDC} For any acts $f, g, h\in \mathcal{F}$, event $E \subset S$, and $\beta\in (0, 1]$, 
\begin{equation}
fEh\succsim gEh \quad \text{if and only if} \quad fEh\succsim_{I^\beta_E} gEh.
\end{equation} 
\end{axm}

In Section \ref{alpha-div}, we define a more general proportionality condition that extends Dynamic Consistency to more general information constraints. To do so, we define information constraints that are ``proportional". It turns out that if information constraints are proportional on $E$, then they must be subsets of $\beta$-events on $E$. For example, an information constraint $I$ is proportional to standard event information $I^1_E$ on $E$ if and only if it is an intersection of a $\beta$-event on $E$ and some information constraint that only imposes restrictions on $E^c$. Hence, by imposing a proportionality on $\beta$-events, \nameref{IDC} seems a natural extension of Dynamic Consistency. 

It turns out that within the class of \nameref{iuseu} preferences, \nameref{IDC} fully characterizes \nameref{GBU}.

\begin{prps}\label{DCprop} Let $\{\succsim_I\}_{I\in\mathscr{I}}$ be a family of \nameref{iuseu} preferences. Then, $\{\succsim_I\}_{I\in\mathscr{I}}$ satisfy \nameref{IDC} if and only if beliefs exhibit \nameref{GBU}.\end{prps}

When an information constraint $I^\beta_E$ is revealed, a DM following Extended Bayesian updating allocates $\beta$ among the states in $E$ in a proportional manner and the remaining probability mass, $1-\beta$, among the states in the complementary event $E^c$.

We now formally show that under some minor technical conditions on the distance function $d_\mu$, \nameref{IDC} essentially characterizes the family of  \nameref{GBex} distance functions.

\begin{thm}\label{bayesdistprop} Suppose $|S|\ge 3$ and $d_\mu(\pi)=\sum^n_{i=1} d_i(\pi_i, \mu_i)$ where each $d_i(\cdot, \mu_i)$ is differentiable and strictly convex. If \nameref{IDC} is satisfied, then there is a strictly concave $\sigma$ such that \[d_\mu(\pi)=-\sum^n_{i=1} \mu_i\,\sigma(\frac{\pi_i}{\mu_i})\text{ for }\pi\in \Delta\text{ with }\pi_i\le\frac{\mu_i}{\mu_i+\min_{j} \mu_j}\text{ for each }i.\]\end{thm}

The restriction $\pi_i\le \frac{\mu_i}{\mu_i+\min_{j} \mu_j}$ is simply due to the limitation that we observe a single prior $\mu$. If IU beliefs exhibit \nameref{GBU} for all possible priors, the restriction can be dropped since $\frac{\mu_i}{\mu_i+\min_j \mu_j}$ can be arbitrarily close to $1$.

The above results provide behavioral foundations for $f$-divergences, the most popular class of divergences in information theory. 

\subsection{Characterizing $\alpha$-divergence.}\label{alpha-div}

Our characterization of \nameref{GBex} shows that we cannot distinguish between Bayesian divergences using $\beta$-events, which include standard event information. In fact, Bayesian divergences cannot be distinguished using interval information constraints as the following proposition shows.

\begin{prps}\label{prop:interval} For any Bayesian divergence, the posterior after interval information $I=\{\pi\in\Delta(S)\mid \underline{\beta}\le \pi(E)\le \overline{\beta}\}$ can be characterized as follows:
\[\mu_I=\mu_{I^\gamma_E}=\arg\min d^{KL}_\mu(I)\]
where $\gamma=\arg\min_{x\in [\underline{\beta}, \overline{\beta}]}|x-\mu(E)|$.
\end{prps}

The proposition shows that for all interval information constraints,  the posterior $\mu_I$ is independent of the functional form of $\sigma$ in Bayesian divergences. Hence, we require richer forms of information to distinguish between various Bayesian divergences.  In the remainder of this section, we show that we can not only distinguish between the various Bayesian divergences, but we can behaviorally characterize $\alpha$-divergences and the Kullback-Leibler divergence.

Our first result shows that we are able to characterize a family of divergences known as $\alpha$-divergences (equivalently, R\'enyi divergences) by strengthening \nameref{IDC}. This family includes KL divergence as the limiting case of $\alpha=0$. To do so, we introduce the following axiom. Because some $\alpha$-divergences have an unbounded boundary, we characterize $\alpha$-divergences on the interior of $\Delta(S)$.

\nameref{IDC} imposes a certain proportionality condition on beliefs for $\beta$-events for a fixed event $E$. Here we strengthen this axiom by imposing the same proportionality property on information constraints that are scaled versions of each other on a fixed event $E$. To define this extension formally, we need a few definitions. For any information constraint $I$ and event $E$, we write
\[I(E):=\{(\pi(s))_{s\in E}\in [0, 1]^{E}\,\big{|}\,\exists\pi'\in I\text{ such that }(\pi(s))_{s\in E}=(\pi'(s))_{s\in E}\},\]
which captures the effective informational constraint on $E$. 
We say information constraints $I$ and $I'$ are \textbf{proportional on E} if there is some $\lambda>0$ such that $I(E)=\lambda\, I'(E)$ and $I=I(E)\times I(E^c)$ and $I'=I'(E)\times I'(E^c)$. Proportionality requires two things. First, the informational requirements of $I$ on $E$ can be obtained from $I'$ by a common scale parameter $\lambda$. Second, the informational requirements on $E$ and $E^c$ are sufficiently separable so that the requirements on $E^c$ do not interfere with the requirements on $E$. For example, standard event information $I^1_E$ and the $\beta$-event on $E$ are proportional on $E$.  

\begin{axm}[\bf{Scale Invariance}]\label{SI} For any event $E$, if convex information constraints $I$ and $I'$ are proportional on $E$, then 
\[p\,\{s\}\,r\sim_{I} q\,\{s'\}\,r\text{ iff } p\,\{s\}\,r\sim_{I'} q\,\{s'\}\,r\text{ for any }p, q, r\in\Delta(X)\text{ and }s, s'\in E.\]
\end{axm}

\nameref{SI} ensures that we can shift information constraints on $E$ by a common factor $\lambda$ and the DM's belief selection on $E$ is consistent.

\begin{thm}\label{prop:renyi} Suppose $|S|\ge 4$. Consider a Bayesian divergence $d^\sigma_\mu(\pi)$ with continuously differentiable $\sigma$ on $(0, \infty)$. \nameref{SI} is satisfied for each $\mu$ if and only if $d_\mu$ is an $\alpha$-divergence.\end{thm}

\subsection{Characterizing KL divergence}

Another natural form of information that has appeared in the literature is qualitative information: ``event $A$ is more likely than event $B$." A unique property of the KL divergence (among $\alpha$-divergences) is that such qualitative information does not affect beliefs on $(A\cup B)^c$. We state this property in behavioral terms.

\begin{axm}[\bf{Independence of Irrelevant Events}]\label{IIE} For any information constraint $I=\{\pi\in\Delta(S)\mid\pi(A)\ge \pi(B)\}$ for disjoint events $A$ and $B$, any $p,q\in\Delta(X)$, and $f \in \mathcal{F}$,
\[p\,A\cup B\,f\succsim q \text{ iff }p\,A\cup B\,f\succsim_I q.\]
\end{axm}

It can be shown that \nameref{IIE} alone is not enough to uniquely characterize the KL divergence. However, the KL divergence can be uniquely obtained if \nameref{SI} is also imposed or \nameref{IIE} is strengthened as follows.

\begin{axm}[\bf{Strong Independence of Irrelevant Events}]\label{SIIE} For any information constraint $I=\{\pi\in\Delta(S)\mid\pi(A)\ge \rho\,\pi(B)\}$ for disjoint events $A$ and $B$, any $p,q\in\Delta(X)$, and $f \in \mathcal{F}$,
\[p\,A\cup B\,f\succsim q \text{ iff }p\,A\cup B\,f\succsim_I q.\]
\end{axm}

We can now characterize the KL divergence. Because the KL divergence is unbounded on the boundary, we characterize it on the interior of $\Delta(S)$. 

\begin{prps}\label{prop:kl} Suppose $|S|\ge 4$. Consider a Bayesian divergence $d_\mu(\pi)$. The following are equivalent:
\begin{itemize}
\item[(i)] \nameref{SI} and \nameref{IIE} are satisfied for each $\mu$;
\item[(ii)] \nameref{SIIE} is satisfied for each $\mu$;
\item[(iii)] $d_\mu$ is a KL divergence.
\end{itemize}

\end{prps}

The above result assumes $|S|\ge 4$. However, the characterization based on \nameref{SIIE} only requires $|S|\ge 3$. 

\cite{zhao2022pseudo} axiomatically characterizes KL divergences using information constraints that are of the restrictive form $I=\{\pi\in\Delta(S)\mid\pi(A)\ge \pi(B)\}$. His axioms and our axioms are not directly comparable because, for example, he works with richer primitives and an infinite state space. We also note that \nameref{IIE}, which uses information constraints that are relevant to \cite{zhao2022pseudo}, alone is not enough to uniquely characterize the KL divergence.

\section{Nesting the Bayesian Framework with a Signal Structure}\label{signal}

While our model is intended to capture scenarios where the DM does not know the DGP, it turns out that it nests the standard Bayesian framework where the DGP is known. To show this claim formally, let us consider any collection $M$ of messages $m$ and a joint distribution $P$ over $S\times M$ with marginal $\mu$ on $S$, which represents the DGP or information structure. As usual, $P(m|s)$ is the conditional probability of receiving message $m$ given state $s$. The Bayesian posterior given this DGP, denoted by $\mu^{\text{BU}}$, is calculated in the usual way: 
\[\mu^{\text{BU}}(s|m)=\frac{P(m|s)\mu(s)}{\sum_{s'\in S} P(m|s')\mu(s')}.\]

As one might anticipate, the above Bayesian posterior can be obtained from IU when the information constraint is linear, and the distance function is a \nameref{GBex}. Hence, consider any Bayesian divergence $d^\sigma_\mu(\pi)$. Let $r_i=\frac{P(m|s_i)}{P(m)}$ with $P(m)=\sum_{s_i\in S}\mu_i\, P(m|s_i)>0$ and $\kappa^\sigma_m=\sum_{s_i\in S} \mu_i\,r_i\,\sigma'\big(r_i)$.\footnote{Throughout this section, we assume either $P(m|s_i)>0$ for each $s_i$ or $|\sigma'(0)|<\infty$.} Let us define the following linear information constraint:
\[I^m_\sigma=\{\pi\in\Delta(S)| \sum_{s_i\in S} \sigma'\big(r_i\big)\,\pi_i=\kappa^\sigma_m\},\]
which will correspond to receiving $m$ given the above DGP.

\begin{prps}\label{NestBU} For any Bayesian divergence $d^\sigma_\mu(\pi)$, 
\[\mu_{I^m_\sigma}=\mu^{BU}(\cdot|m).\]
\end{prps}

Our result extends to multiple information structures, with or without a known mixture. In fact, such results immediately follow from Proposition 7. To illustrate, consider joint distributions $P$ and $P'$ over $S\times M$ with common marginal $\mu$ on $S$, and
the DM assigns probability $\lambda\in[0,1]$ to $P$ and $1-\lambda$ to
$P'$. Let
\[
r_i^\lambda
=
\frac{\lambda P(m\mid s_i)
+(1-\lambda)P'(m\mid s_i)}{\lambda P(m)+(1-\lambda)P'(m)}
\text{ and }
\kappa_{m,\lambda}^\sigma
=
\sum_{i=1}^n
\mu_i r_i^\lambda
\sigma'(r_i^\lambda).
\]
Then define the linear information constraint
\[
I_{m,\lambda}^\sigma
=
\left\{
\pi\in\Delta(S)
\;\middle|\;
\sum_{i=1}^n
\sigma'(r_i^\lambda)\pi_i
=
\kappa_{m,\lambda}^\sigma
\right\}.
\]
\begin{cor} For any Bayesian divergence $d_\mu^\sigma$, $\mu_{I_{m,\lambda}^\sigma}(s_i)$ is the Bayesian posterior under the mixture information structure $P_\lambda=\lambda P+(1-\lambda)P'$, i.e., 
\[
\mu_{I_{m,\lambda}^\sigma}(s_i)
=
\frac{
\mu_i\left[
\lambda P(m\mid s_i)
+
(1-\lambda)P'(m\mid s_i)
\right]
}{
\lambda P(m)+(1-\lambda)P'(m)
}.
\]
\end{cor}

Now suppose the DM knows that the information structure is a mixture of
$P$ and $P'$, but does not know the mixture weight $\lambda\in[0,1]$. Let $r_i=\frac{P(m\mid s_i)}{P(m)}$ and $r_i'=\frac{P'(m\mid s_i)}{P'(m)}$. Then define the information constraint
\[
I_m^{P,P'}
=
\left\{
\pi(t)\in\Delta(S)
\;\middle|\;
\pi_i(t)
=
\mu_i\left[
t r_i+(1-t)r_i'
\right],
\quad
t\in[0,1]
\right\}.\]

\begin{cor} For any Bayesian divergence $d_\mu^\sigma$, there is $\lambda^*\in [0, 1]$ such that $\mu_{I_m^{P,P'}}$ is equal to $\mu_{I_{m,\lambda^*}^\sigma}$, the Bayesian posterior under the mixture information structure $P_{\lambda^*}$.
\end{cor}

\section{Bayesian (Dis)Agreement}\label{agreement}

Our definition of \nameref{GBU}, along with its axiomatic characterization and our characterization of \nameref{GBex}, is based on $\beta$-events. Critically, this only pins down a class of distances, rather than a unique distance. Although all \nameref{GBex}s agree on $\beta$-events (they agree on interval information constraints, by Proposition \autoref{prop:interval}), they may disagree for other types of information.  Indeed, our characterizations of $\alpha$-divergences and KL divergence show that other information constraints can be used to distinguish between Bayesian divergences. A consequence of this is that there can be Bayesian disagreement. However, it is important to emphasize that both general information and non-uniqueness of Bayesian divergence are necessary for disagreement. This is because non-uniqueness of Bayesian divergences does not imply disagreement under a standard informational environment.

For additional intuition as to why disagreement arises, consider any \nameref{GBex}.  The properties of $\sigma$ ensure that for any $\beta$-event, the DM always shifts probability mass proportionally across states. Put another way, Dynamic Consistency places restrictions on beliefs after information that is ``event-like," which includes interval information constraints. Crucially, such information does not fundamentally challenge the DM's beliefs about the relative likelihood of states. However, information that is not ``event-like'' may necessitate the revision of relative likelihoods by the DM. When information constraints preclude proportional shifts, the specifics of the distance function matter. 

Formally, we say that two distance functions $d_\mu$ and $d'_\mu$ are \emph{distinct} if there are no constants $\alpha_\mu>0$ and $\beta_\mu$ such that $d_\mu(\pi)=\alpha_\mu d'_\mu(\pi)+\beta_\mu$ for each $\pi$.

\begin{prps}\label{prop:disagree} Suppose either $|S|\ge 4$ or $|S|=3$ and $\mu$ is non-uniform. Take any two distinct Bayesian divergences $d_\mu$ and $d'_\mu$ generated by twice-differentiable functions $\sigma$ and $\delta$ on $(0, \infty)$. Then there is an $I$ such that $\mu_I\neq\mu'_I$.
\end{prps}

The result above may also be interpreted as a uniqueness result for Bayesian divergences. Here we obtain a stronger uniqueness result than Proposition \ref{uniqueness} because the distance functions are assumed to be Bayesian. 

While Proposition \ref{prop:disagree} establishes the existence of disagreement in general, it does not directly establish conditions on distributions. For the next proposition, we provide a necessary and sufficient condition on distributions to ensure robust agreement for Bayesian divergences.

Now consider two DMs with a common prior $\mu$ and two potential posteriors, $\pi,\pi'$, all of which have full support. A row-stochastic matrix $K$ on $S$ is an $|S| \times |S|$ matrix with non-negative entries whose rows sum to one. We can think of $K$ as a noisy reporting device. When $\mu\, K=\mu$, this reporting device does not distort the prior. When $\pi=\pi' K$, we can think of $\pi$ as a garbling of $\pi'$ by $K$. When such a $K$ exists, it turns out that every Bayesian divergence will agree that $\pi$ is closer to $\mu$ than $\pi'$.

\begin{prps}\label{prop:blackwell-agreement}
Let $\mu,\pi,\pi'\in\operatorname{int}\Delta(S)$. The following
statements are equivalent:
\begin{enumerate}
    \item $d_\mu(\pi) \leq d_\mu(\pi')$ for every Bayesian divergence.
\item There exists a row-stochastic matrix $K$ such that $\mu=\mu\, K$ and $\pi=\pi' K$.

\end{enumerate}
\end{prps}

A corollary is that if no row-stochastic matrix $K$ satisfying $\mu K=\mu$ satisfies either $\pi'K=\pi$ or $\pi\, K=\pi'$, so that neither distribution is a prior-preserving garbling of the other, then there exist Bayesian divergences $d_\mu^\sigma$ and $d_\mu^\delta$ such that $d_\mu^\sigma(\pi)<d_\mu^\sigma(\pi')$ and $d_\mu^\delta(\pi)>d_\mu^\delta(\pi')$.

For example, let $\mu=(.5,.25,.25)$, $\pi=(.4,.2,.4)$, and $\pi'=(.6,.3,.1)$. The ratios of posterior to prior probabilities, $\pi_i/\mu_i$, are $(.8,.8,1.6)$ for $\pi$ and $(1.2,1.2,.4)$ for $\pi'$. A prior-preserving garbling averages these ratios. Thus $\pi'K=\pi$ is impossible because $1.6>1.2$, while $\pi K=\pi'$ is impossible because $.4<.8$. Hence there exist Bayesian divergences that rank 
$\pi$ and $\pi'$ strictly oppositely and hence there is disagreement at $\{\pi,\pi'\}$.

This result also has an interpretation in terms of Blackwell's ordering of experiments. For each candidate posterior $\pi$, consider a hypothetical binary experiment $\mathcal{E}_{\mu}^{\pi}$ with hypotheses $\theta_0: \mu \text{ is the true model}$ and $\theta_1: \pi \text{ is the true model}$. The experiment $\mathcal{E}_{\mu}^{\pi}$ returns a signal $z_{s_i}$ with probability $\mu(s_i)$ when $\theta_0$ is true and $\pi(s_i)$ when $\theta_1$ is true. Condition (2) of Proposition \ref{prop:blackwell-agreement} is equivalent to $\mathcal E_\mu^\pi$ is a garbling of $\mathcal E_\mu^{\pi'}$, since the same matrix $K$ transforms the signal distributions under both hypotheses. Thus all Bayesian divergences regard $\pi$ as weakly closer to $\mu$ exactly when the experiment distinguishing $\pi$ from $\mu$ is Blackwell less informative than the experiment distinguishing $\pi'$ from $\mu$. Blackwell incomparability, therefore, permits strictly opposite rankings and hence disagreement after the common information constraint $\{\pi,\pi'\}$.

The previous proposition establishes conditions under which all Bayesian divergences agree, or disagree, on the order of two possible beliefs. The result in the next subsection shows that disagreement can extend to convex information constraints, rather than just finite ones.

\subsection{Disagreement and Comparative Statics}

Now that the existence of disagreement has been established, we turn to the problem of understanding \emph{which information constraints induce disagreement} and \emph{in which direction disagreement occurs}. To do so, we establish a result showing that information constraints of the form $I=\{\pi(A)=\pi(B)\}$ induce disagreement and describing how the concavity of $\sigma$ affects the resulting disagreement. To illustrate, consider two DMs with common prior $\mu$ and Bayesian divergences with twice differentiable $\sigma^1$ and $\sigma^2$ such that $(\sigma^2)'$ is strictly more concave than $(\sigma^1)'$, i.e., $(\sigma^2)'=f\circ (\sigma^1)'$ for some strictly concave $f$. For any information constraint $I$, let us denote DM $k$'s posteriors by $\mu^k_I$. 

\begin{prps}\label{compstat} Suppose $|S|\!\ge\!3$ and $I\!=\{\pi\in\Delta(S)|\pi(A)\!=\!\pi(B)\}$ for any partition $\{A, B, C\}$ of $S$. If $\mu\not\in I$, then $\mu^1_I(s)>\mu^2_I(s)\text{ for any }s\in A\cup B$, and $\mu^1_I(s)<\mu^2_I(s)\text{ for any }s\in C$.

\end{prps}

Proposition \ref{compstat} shows that a DM with a more concave $\sigma'$ puts less posterior weight on $A\cup B$ and more weight on $C$. Since the information requires $A$ and $B$ to receive equal posterior probability despite having unequal prior probabilities, updating necessarily involves different proportional adjustments to $A$ and $B$. Greater concavity of $\sigma'$ makes this asymmetry more consequential, leading the DM to reduce the total probability assigned to the constrained events $A$ and $B$ and shift more probability toward the unconstrained event $C$.

\subsection{Polarization}

Proposition \ref{prop:disagree} shows that common priors and
Bayesian updating on $\beta$-events do not guarantee agreement
under general information. To illustrate, suppose two policymakers
share the prior $\mu=(.4,.3,.3)$ and a scientific panel reports
that the true distribution belongs to
$\{\pi,\pi'\}$, where
$\pi=(.3,.375,.325)$ and $\pi'=(.335,.405,.26)$.
Alice uses the KL divergence, while Bob uses the Bayesian
divergence generated by $\sigma(x)=\sqrt{x}-1$.
Alice selects $\pi'$, whereas Bob selects $\pi$.\footnote{
The KL distances to $\pi$ and $\pi'$ are approximately
$.024117$ and $.023832$, respectively; Bob's distances are
$0.005930$ and $0.006087$.} Thus the same public information induces different beliefs, although both policymakers would agree after every $\beta$-event (see figure \ref{disagreefig} for an illustration).

Hence, Bayesian agents with a common prior disagree after public information. Aumann's agreement theorem does not hold because of general information and non-uniqueness of Bayesian divergences.

\begin{figure}
	\centering
     \includegraphics[width=6.5cm]{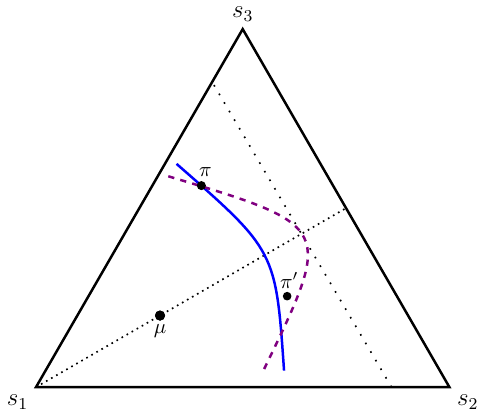}

\caption{Bayesian disagreement at $\pi$ and $\pi'$. Iso-distance curves with respect to the common prior $\mu$ are depicted for two DMs. The solid (blue) curve represents DM 1, while the dashed (purple) curve represents DM 2.}\label{disagreefig}
          \end{figure}

\subsection{Trade and Public Information}

Under the standard assumptions of no-trade theorems, Bayesian
traders with a common prior cannot use public or private information
to generate mutually acceptable trade from an ex ante
Pareto-efficient allocation \citep{Milgrom1982, Morris1994}.
For traders using Bayesian divergences, full insurance remains
efficient after public $\beta$-events and, by
Proposition \ref{prop:interval}, interval information.
More general public information can instead induce disagreement
and speculative trade.

To illustrate, take Alice and Bob from the preceding polarization
example. An allocation $f=(f^A,f^B)$ is a tuple of state-contingent consumption of one commodity (i.e., $f^i \in \mathbb{R}_{+}^3$ with $i \in \{A,B\}$).  Both traders are SEU maximizers with respect to the same (strictly) concave utility function, $u^A(x)=u^B(x)=\sqrt{x}$ for any $x \in \mathbb{R}_{+}$. For concreteness, suppose each has initial endowment
$e^A=e^B=(5,5,5)$. After learning $I=\{\pi,\pi'\}$, their posteriors
are $\mu_I^A=\pi'$ and $\mu_I^B=\pi$, respectively. In particular,
Alice assigns probability $.405$ to $s_2$, while Bob assigns
probability $.375$. 

Since both traders disagree on their posteriors,  a Pareto improving exchange is possible. For instance, the feasible allocation $f^A=(4.81,5.30,4.81)$ and $f^B=(5.19,4.70,5.19)$ makes both traders strictly better off than the full-insurance allocation.\footnote{$0.335\sqrt{4.81}+0.405\sqrt{5.30}+0.26\sqrt{4.81}>\sqrt{5}$
and $0.3\sqrt{5.19}+0.375\sqrt{4.70}+0.325\sqrt{5.19}>\sqrt{5}$.} Notice that the trade leading to the Pareto-superior allocation $(f^A,f^B)$  is not driven by risk sharing but by speculative motives. Put differently, both traders are willing to abandon the full-insurance allocation in order to bet against each other by purchasing assets that correspond to the transfers $(f^A-e^A)$ and $(f^B-e^B)$. \cite{Gilboa2014} call such trades ``speculative'' Pareto improvement bets.\footnote{\cite{Gilboa2014} distinguish between Pareto improvements due to betting and due to risk-sharing. An allocation $f$ is called a \textit{bet} if $f$ Pareto dominates another allocation $g$ with the full-insurance property.}

Hence, no-trade theorems break because of Bayesian disagreement in the general information setting.

\section{Related Literature}\label{literature}

There are a few papers that develop a theory of belief updating for general information structures. Perhaps the first to study general information was \cite{damiano2006}, who shows the nonexistence of belief selection rules under three properties. In a similar spirit to \cite{damiano2006}, \cite{chambers2010} prove the nonexistence of a selection rule satisfying a particular path-independence condition that they call Bayesian consistency.  In contrast, we focus on the behavioral implications of a particular selection rule: distance minimization. 

As far as we are aware, the first paper to combine general information and distance minimization is \cite{zhao2022pseudo}, who considers environments in which a DM receives a sequence of qualitative statements of the form ``event $A$ is more likely than event $B$." This DM has a probabilistic belief and updates it via the so-called Pseudo-Bayesian updating rule.\footnote{The pseudo-Bayesian updating rule is axiomatized by two axioms, Exchangeability and Stationarity, directly imposed on posteriors. Exchangeability requires that the order of information does not matter as long as the DM receives qualitative statements that ``neither reinforce nor contradict each other.'' Stationarity requires that the DM's beliefs do not change when a qualitative statement is consistent with the prior.} There are two key differences between \cite{zhao2022pseudo} and our paper. First, he focuses on the Kullback-Leibler divergence, while we allow for general distances. Thus, our framework allows for Bayesian disagreement. Second, he focuses on qualitative information (``$A$ is more likely than $B$''), while we allow much more general information. General forms of information allow us to characterize $f$-divergences and $\alpha$-divergences (equivalently, R\'enyi divergences), and to show that Bayesian updating with a known DGP is nested by IU. 

More behaviorally, our axiomatic characterizations of $\alpha$-divergences and the KL divergence show that IU allows deviations from \nameref{SI} and \nameref{SIIE}. Moreover, wIU allows deviations from \nameref{COMP} and partial trust in qualitative (AI) recommendations, experimentally supported in \cite*{kovach2026recommendations}.  

\cite*{ke2024learning} develop and axiomatize the contraction rule (CR), a belief updating rule for general information. In CR, the DM mixes the prior and a representative belief $\pi_I\in I$ for the information constraint $I$ with an information-dependent weight $w(I)$, i.e., $\mu_I=w(I)\mu+(1-w(I))\pi_I$. Let us now compare CR and wIU.

In CR, $\pi_I$ is independent of the prior but can be any element of $I$. Consequently, $\mu_I$ is linear in the prior for a given information constraint $I$, but it can be any element of the convex hull of ${\mu}\cup I$. In other words, the relationship between posterior and information is very flexible in CR, while that between posterior and prior is very restrictive. CR does not allow Bayes' rule as a special case, as it violates \nameref{IDC}.

In contrast, in wIU, weight $w$ is independent of $I$, which is behaviorally captured by \nameref{Lin} and \nameref{EInd}. Moreover, $\pi_I(\mu)$ minimizes the distance between $\mu$ and $I$. Hence, wIU has sharp predictions even for a single $I$. For example, when the information confirms the DM's prior, i.e., $\mu\in I$, there is no updating, i.e., $\mu_I=\mu$. However, in CR, $\mu_I$ can be any element of $I$. In the experimental work of \cite*{kovach2026recommendations}), they document that most subjects don't update their beliefs when information confirms their priors, consistent with the prediction of wIU.

Other papers that feature general information include  \cite{gajdos2008}, \cite{Dominiak2020}, as well as \cite{ok2022believing}. These papers study the embedding of general information under more general preferences than SEU. In \cite{ok2022believing}, an agent ranks ``info-acts'' $(\mu, f)$, which consist of a probability distribution $\mu$ over $S$ and a (Savage) act $f$. They characterize the notion of probabilistically amenable preferences: the agent adopts  $\mu$ as her own belief, and evaluates act $f$ via the lottery induced by $\mu$ over the outcomes of $f$, although lotteries may not be evaluated by expected utility.\footnote{Probabilistically amenable preferences are weaker than the probabilistically sophisticated preferences introduced by \cite{machina1992probabilisticsophistication} as they do not need to be complete, continuous, and consistent with first-order stochastic dominance.} In their setting, each act might be evaluated with respect to a different probability distribution.  In IU, \nameref{COMP} ensures that when $I=\{\mu\}$, the DM adopts $\mu$ to evaluate all acts.

\cite{gajdos2008} studies preferences defined over more general ``info-acts'' $(P, f)$ where  $P \subseteq \Delta(S)$. They characterize a representation in which an ambiguity-averse agent selects a set of priors $\varphi(P) \subseteq \Delta(S)$ to evaluate $f$ via the Maxmin-criterion, and derive conditions for $\varphi(P) \subseteq P$. For the case in which $P$ is a $\beta$-event, \cite{Dominiak2020} show that whether an agent selects priors from the objective set $\Delta(S)$ or restricts attention to the exogenous set $P$ can fundamentally affect preferences, including whether the resulting preferences are consistent with the puzzles in \cite{machina2009risk}

The literature on updating with standard information, i.e., events, is much larger. Within this environment, a few papers, \cite{Perea2009}, \cite{Basu_2018}, \cite{dominiak2023inertial}, have studied inertial or minimum distance updating rules. While these three papers focus on non-Bayesian updating and/or zero-probability events in the standard information environment, the current paper focuses on updating under general information with an unknown DGP, Bayesian disagreement, and proposes and characterizes a notion of Bayesian updating within this general framework. For the connection between Inertial Updating and the theoretical and experimental literature on non-Bayesian belief updating, see \cite{dominiak2023inertial}. See also \cite*{kovach2026recommendations} for the detailed discussion on how our environment with unknown DGP and general information differs from the existing experimental (standard informational) environment for non-Bayesian belief updating.

Although it is not related to the economic contexts we focus on (choice under uncertainty, general information, and belief updating), \cite{csiszar1991least} provides an axiomatic characterization of $f$-divergence when $f$ is strictly convex (which corresponds to strictly concave $\sigma$). Not only are the settings and primitives different, but our paper behaviorally characterizes $f$-divergence by generalizing Dynamic Consistency for general information. 

Several papers have utilized divergences to generalize the maxmin representation of \cite{gilboa1989maxmin} for ambiguity-averse preferences. \cite{maccheroni2006ambiguity} introduce and characterize the variational representation of ambiguity-averse preferences in which the utility of act $g$ is $V(g)=\min_{\pi\in \Delta(S)} \{\sum_{s\in S} \pi(s) u(g(s))+c(\pi)\}$. They also introduce divergence preferences, a special case of variational preferences in which $c$ is an $f$-divergence, but they do not characterize divergence preferences. \cite{strzalecki2011axiomatic} behaviorally characterizes multiplier preferences where $c$ is a KL divergence. The variational representation is conceptually different from \nameref{iuseu} since the effective subjective belief (or the minimizer $\pi$ of  $\sum_{s\in S} \pi(s) u(g(s))+c(\pi)$) is act-dependent and $c$ becomes redundant for SEU preferences. Similarly, \citet{cerreia2026making} consider preferences conditional on sets of models $Q$ and characterize divergence preferences, with the goal of separating ambiguity aversion from concerns for model misspecifications. Their DM can be represented by $V(g)=\min_{\pi\in \Delta(S)} \{\sum_{s\in S} \pi(s) u(g(s))+ \min_{q \in Q}c(\pi, q)\}$. Behaviorally, this DM is similar to one with variational preferences, but she faces a greater penalty for using $\pi \notin Q$. Like with variational preferences, $c$ becomes redundant for SEU preferences. Moreover, the effective subjective belief (or the minimizer $\pi$) and the effective mental model (or the minimizer $q$) are both act-dependent. In addition, their DM does not need to select from $Q$, and therefore $Q$ is a ``soft constraint,'' whereas our IU DM always selects a belief from the information constraint $I$. Neither of the above papers is concerned with belief updating.



\appendix
\section{Proofs}\label{proofs:main}

\subsection{Proof of Proposition \ref{prop:bayes}}

Take any \nameref{iuseu} representation with \nameref{GBex} $d_\mu(\pi)=-\sum^n_{i=1}\mu_i\, \sigma(\frac{\pi_i}{\mu_i})$. Without loss of generality, let $E=\{s_1, \ldots, s_m\}$ and $\beta\in (0, 1]$. By the representation, we should find $(\pi_1, \ldots, \pi_m)$ to solve

\[\max \sum^m_{i=1}\mu_i\, \sigma(\frac{\pi_i}{\mu_i})\text{ subject to} \sum^m_{i=1} \pi_i=\beta\text{ and }\pi_i\ge 0.\]

Since $\sigma$ is strictly concave, by Jensen's inequality, 
\[\sum^m_{i=1}\frac{\mu_i}{\mu(E)}\, \sigma(\frac{\pi_i}{\mu_i})\le \sigma\big(\sum^m_{i=1}\frac{\mu_i}{\mu(E)}\frac{\pi_i}{\mu_i}\big)=\sigma(\frac{\beta}{\mu(E)}),\]
where the equality is obtained at $(\pi_i)^m_{i=1}=(\beta\,\frac{\mu_i}{\mu(E)})^m_{i=1}$. Moreover, since $\sigma$ is strictly concave, the above has a unique solution, which is  $(\pi_i)^m_{i=1}=(\beta\,\frac{\mu_i}{\mu(E)})^m_{i=1}$.

\subsection{Proof of Proposition \ref{prop:anything}}

Take any minimally responsive $\rho_\mu$. Suppose that we can find $\lambda(I)>0$ such that 
\[\mu(s)=\sum_{I}\lambda(I)\rho_\mu(I)(s)\text{ for each }s\in S.\]
Then let $P(I|s):=\lambda(I)\,\frac{\rho_\mu(I)(s)}{\mu(s)}$. Note that $P(\cdot|s)$ is a probability distribution since
\[\sum_{I}P(I|s)=\sum_{I}\lambda(I)\,\frac{\rho_\mu(I)(s)}{\mu(s)}=\frac{\sum_{I}\lambda(I)\,\rho_\mu(I)(s)}{\mu(s)}=1.\]
Since $\sum_{s\in S}\rho_\mu(I)(s)=1$, we obtain $\lambda(I)=\sum_{s\in S}P(I|s)\mu(s)$. Hence, $\rho_\mu$ has a Bayesian representation. We must now show that $\lambda$ exists. 

Since $\mathcal I$ is finite, the existence of strictly positive $\lambda$ is equivalent to $\mu$ lying in the relative interior of $\text{conv}\{\rho_\mu(I):I\in\mathcal{I}\}\subseteq \{x\in\mathbb R^{|S|}:\sum_s x(s)=1\}$. Suppose not. Then by a standard separating hyperplane theorem there exists a nonconstant $v$ such that $v\cdot\rho_\mu(I)\leq v\cdot\mu$ for every $I$. Let $\bar{v}=\max_s v(s)$ and $E = \{s \in S: v(s) = \bar{v}\},$ which is nonempty and a proper subset of $S$. Then $v \cdot (\rho_\mu(I_E) - \mu) = \sum_{s \in E} \bar{v} (\rho_\mu(I_E)(s) - \mu(s))  + \sum_{s \in E^c}v(s) (\rho_\mu(I_E)(s) - \mu(s))$. Since $\rho_\mu(I_E)$ and $\mu$ must sum to one, it follows that $\sum_{s \in E}(\rho_\mu(I_E)(s) - \mu(s)) = -\sum_{s \in E^C}(\rho_\mu(I_E)(s) - \mu(s))$, and so 
\begin{align*}v \cdot (\rho_\mu(I_E) - \mu) &= - \bar{v} \sum_{s \in E^c} (\rho_\mu(I_E)(s) - \mu(s))  + \sum_{s \in E^c}v(s) (\rho_\mu(I_E)(s) - \mu(s)) \\
 & =  \sum_{s \in E^c} (v(s) - \bar{v} )(\rho_\mu(I_E)(s) - \mu(s))  > 0.
\end{align*}
The last inequality follows from the fact that $v$ is nonconstant and Minimal Responsiveness ensures $\rho_\mu(I_E)(s)<\mu(s)$ for all $s\in E^c$. This contradicts $v\cdot\rho_\mu(I)\leq v\cdot\mu$, hence $\mu$ must be in the relative interior of  $\text{conv} \{\rho_\mu(I):I\in\mathcal{I}\}$, hence there exist positive weights $\lambda(I)$ such that  $\mu=\sum_{I\in\mathcal I}\lambda(I)\rho_\mu(I)$ and $\sum_{I\in\mathcal I}\lambda(I)=1$.

\subsection{Proof of Theorem \ref{repthrm2}}

\noindent\textbf{Sufficiency.} We first prove sufficiency.  Consider a family of preferences $\{\succsim_I\}_{I\in\mathscr{I}}$ that satisfies \nameref{SEU}, \nameref{CON},  \nameref{IB}, \nameref{EXTM}, \nameref{RCONT}, \nameref{RCON}, \nameref{Lin}, \nameref{EInd}, and \nameref{IComp}.

\bigskip
\noindent\textbf{Step 1.} Under \nameref{SEU}, there is $(u, (\mu_{I})_{I\in\mathscr{I}})$ such that for each $I\in\mathscr{I}$, the conditional preference relation $\succsim_I$  admits the SEU representation with $(u, \mu_I)$. As usual, $u$ is unique up to a positive affine transformation, and $\mu_I$ is unique. Let $\mu=\mu_{\Delta(S)}$.

\bigskip
\noindent\textbf{Step 2.} There is $\alpha\in[0, 1]$ such that $\mu_\pi=\alpha\,\pi +(1-\alpha) \mu$ for any $\pi\in\Delta(S)$.\medskip
 
 First, consider the trivial case in which $\succsim=\succsim_I$ for all $I \in \mathscr{I}.$ Then clearly the above holds for $\alpha=0$ and, further, the representation holds for $w=1$ and arbitrary admissible $(d_\mu,\succ_{d_{\mu}})$. Henceforth assume that $\succsim_I \neq \succsim$ for some $I \in \mathscr{I}$. Then by \nameref{CON}, there is some $\pi' \in I$ so that $\succsim_{\{\pi'\}}=\succsim_I\neq \succsim$.  

Without loss of generality, we can write $\mu_\pi(s_i)=G_i(\pi)$ for some function $G_i:\Delta(S)\to [0, 1]$ such that $\sum_{s_i\in S} G_i(\pi)=1$. \nameref{EInd} implies that for any $\pi, \pi'$ and $E\subset S$, $\mu_\pi(E)=\mu_{\pi'}(E)$ whenever $\pi(E)=\pi'(E)$. By setting $E=\{s_i\}$, we obtain $G_i(\pi)=G_i(\pi')$ whenever $\pi_i=\pi'_i$. Let $f_i:[0, 1]\to [0, 1]$ be a function such that $f_i(\pi_i)=G_i(\frac{1-\pi_i}{n-1}, \ldots, \pi_i, \ldots, \frac{1-\pi_i}{n-1})$. Then by \nameref{EInd}, $f_i(\pi_i)=G_i(\pi)$ for any $\pi$. 

Take any $\pi, \pi'$, and $E\subset S$. We can find $p, q, r, r'\in\Delta(X)$ such that $p \succ q$,  $p\,E\,q\sim_{\{\pi\}}\,r$ and $p\,E\,q\sim_{\{\pi'\}}\,r'$, i.e.,  $\mu_\pi(E)=\frac{u(r)-u(q)}{u(p)-u(q)}$ and $\mu_{\pi'}(E)=\frac{u(r')-u(q)}{u(p)-u(q)}$. Then by \nameref{Lin} we have 
$p\,E\,q\sim_{\{\lambda\,\pi+(1-\lambda)\,\pi'\}}\,\lambda\, r+(1-\lambda)\,r'$, i.e., 
\[\mu_{\lambda\,\pi+(1-\lambda)\,\pi'}(E)=\frac{\lambda\,u(r)+(1-\lambda)\,u(r')-u(q)}{u(p)-u(q)}.\]
Hence, we obtain 
\[\mu_{\lambda\,\pi+(1-\lambda)\,\pi'}(E)=\lambda\,\mu_\pi(E)+(1-\lambda)\,\mu_{\pi'}(E).\]
By setting $E=\{s_i\}$, we obtain 
\[f_i(\lambda\,\pi_i+(1-\lambda)\, \pi'_i)=\lambda\,f_i(\pi_i)+(1-\lambda)\,f_i(\pi_i').\]
By setting $\pi_i=1$ and $\pi'_i=0$, we obtain 
\[f_i(\lambda)=\lambda\,f_i(1)+(1-\lambda)f_i(0)=\lambda (f_i(1)-f_i(0))+f_i(0).\]
Further, \nameref{EInd} implies $f$ is weakly increasing, hence, $f_i(x)=\alpha_i\, x+\beta_i$ for some $\alpha_i, \beta_i\ge 0$. Since $\sum_{s_i\in S} f_i(\pi_i)=1$ for any $\pi$, we obtain 
\[1=\alpha_i+\sum_j \beta_j\]
by setting $\pi_i=1$. Hence, $\alpha=\alpha_i$ for some $\alpha \ge 0$. Since $\beta_i\ge 0$, we also have $\alpha\le 1$.

By \nameref{CON}, there is $\pi\in\Delta(S)$ such that $\succsim:=\succsim_{\Delta(S)}=\succsim_{\{\pi\}}$. By \nameref{IComp}, we obtain $\pi(E)=\mu(E)$ for each $E$, consequently, $\pi=\mu$. Hence, $\mu_i = \mu_{\mu}(s_i)=\alpha \mu_i + \beta_i$ and therefore $\beta_i=(1-\alpha)\mu_i$ and $\sum_j \beta_j=(1-\alpha)$. Consider any $\pi\in\Delta(S)$ such that $\succsim \neq\succsim_{\{\pi\}}$. Then $\mu\neq\mu_{\pi} = \alpha \pi + (1-\alpha)\mu$, and thus $\alpha > 0$ and $w:=1-\alpha < 1$.

\bigskip
\noindent\textbf{Step 3.} By \nameref{CON} and Step 2, for any $I$, there is $\pi_I\in I$ such that $\mu_I=\alpha\,\pi_I +(1-\alpha) \mu$.

\medskip

We define a mapping $C:\mathscr{I}\to \Delta(S)$ such that for any $I\in \mathscr{I}$, $C(I)=\pi_{I}$ where $\mu_I=\alpha\,\pi_I +(1-\alpha) \mu$. Since $C(I)\in I$, mapping $C$ is a choice function on $\mathscr{I}$. Note that $\succsim_I=\succsim_{I'}$ if and only if $C(I)=C(I')$. Moreover, $C(\Delta(S))=\mu$.

\bigskip
\noindent\textbf{Step 4.} $C$ satisfies Sen's $\alpha$-property. That is, for any $I_1, I'$ with $I'\subset I_1$ and $\pi\in I'$, if $\pi=C(I_1)$, then $\pi=C(I')$. \bigskip

Let $I_2=\{\pi\}$. Note that $I_2\subseteq I'\subseteq I_1$ and $\pi=C(I_1)$ implies that $\succsim_{I_1}=\succsim_{I_2}$. By \nameref{IB}, we have $\succsim_{I_1}=\succsim_{I'}$, equivalently, $\pi=C(I')$.\smallskip

\medskip
\noindent\textbf{Revealed Preference.} Define a binary relation $\succsim^*$ on $\Delta(S)$ as follows: for any $\pi, \pi'\in\Delta(S)$, $\pi'\succ^* \pi$ if $\pi=C(\{\pi, \pi'\})$ and $\pi\neq \pi'$, and $\pi\sim^*\pi'$ if $\pi=\pi'$. Note that $\succsim^*$ is a strict preference relation.

\medskip
\noindent\textbf{Step 5.} Since $C$ is a choice function and satisfies Sen's $\alpha$-property, $\succsim^*$ is a complete and transitive preference relation. Moreover, $C(I)=\min(I, \succsim^*)$. 

\medskip
Since $\mu=\mu_{\Delta(S)}$, we have $\mu=\min(\Delta(S), \succsim^*)$.
\medskip

\medskip
\noindent\textbf{Step 6.} $\pi \succsim^* \pi'$ implies $\pi\, R\, \pi'$. Consequently, since $\succsim^*$ is complete, $R$ is complete.\medskip

Suppose $\pi \succsim^* \pi'$, i.e., $\succsim_{\{\pi, \pi'\}}=\succsim_{\{\pi'\}}$. Note that for any $\epsilon>0$, $\pi\in B(\pi, \epsilon)$ and $\pi'\in B(\pi', \epsilon)$. Hence there is no $\epsilon>0$ such that $\succsim_{\{\pi_\epsilon, \pi'_{\epsilon}\}}=\succsim_{\{\pi_{\epsilon}\}}$ for any $\pi_\epsilon\in B(\pi, \epsilon)$ and $\pi'_\epsilon\in B(\pi', \epsilon)$, i.e., $\pi\, R\, \pi'$.

\medskip
\noindent\textbf{Step 7.} $R$ is continuous.\medskip

Take any two sequences $\pi^n, \tilde{\pi}^n$ such that $\pi^n\, R\, \tilde{\pi}^n$ for each $n$ and $\pi^n\to \pi^*$ and $\tilde{\pi}^n\to \tilde{\pi}^*$. To prove that $\pi^*\, R\, \tilde{\pi}^*$, take any $\epsilon>0$. Since $\pi^n\to \pi^*$ and $\tilde{\pi}^n\to \tilde{\pi}^*$, we can find $n(\epsilon)$ such that for any $n>n(\epsilon)$, $\pi^n\in B(\pi^*, \epsilon/3)$ and $\tilde{\pi}^n\in B(\tilde{\pi}^*, \epsilon/3)$. Moreover, for any $n$, $\pi^n\, R\, \tilde{\pi}^n$ implies we can find $\pi^n_\epsilon\in B(\pi^n, \epsilon/3)$ and $\tilde{\pi}^n_\epsilon\in B(\tilde{\pi}^n, \epsilon/3)$ such that $\pi^n_\epsilon\succ^*\tilde{\pi}^n_\epsilon$. Note that $\pi^n_\epsilon\in B(\pi^*, \epsilon)$ and $\tilde{\pi}^n_\epsilon\in B(\tilde{\pi}^*, \epsilon)$ for any $n>n(\epsilon)$. Hence, for any $\epsilon>0$, we found the desired $\pi^n_\epsilon\in B(\pi^*, \epsilon)$ and $\tilde{\pi}^n_\epsilon\in B(\tilde{\pi}^*, \epsilon)$ such that $\pi^n_\epsilon\succ^*\tilde{\pi}^n_\epsilon$.

\medskip
\noindent\textbf{Step 8.} There is a function $d:\Delta(S)\to \mathbb{R}$ such that $C(I)=\min\big(\arg\min d(I), \succ^*\big)$ where $\succ^*$ is the desired tie-breaker.\smallskip

By Steps 6-7, $R$ is complete and continuous. By \nameref{RCONT}, $R$ is transitive. Hence, by Debreu's theorem, there exists a continuous function $d$ such that $d(\pi)\ge d(\pi')$ if and only if $\pi R \pi'$. Then, by Step 6, we obtain $C(I)=\min\big(\arg\min d(I), \succ^*\big)$. 

Since $C(I)=\min(I, \succ^*)$, $\pi\succ^* C(I)$ for any $\pi\in I\setminus C(I)$. By Step 6, we have $d(C(I))\le d(\pi)$ for any $\pi\in I$. Hence, $C(I)\in \arg\min d(I)$. Therefore, since $C(I)=\min(I, \succ^*)$, $C(I)=\min\big(\arg\min d(I), \succ^*\big)$. Moreover, by Step 6, $d(x)>d(x')$ implies $x\succ^* x'$. Hence, $\succ^*$ is a tie-breaker of $d$.

\medskip
\noindent\textbf{Step 9.} $d$ is a distance function with respect to $\mu$.\medskip

For any $\pi\in\Delta(S)\setminus\{\mu\}$, since $\mu=C(\Delta(S))$ and $[\pi, \mu]\subset \Delta(S)$, we have $\mu=C[\pi, \mu]$. By Revealed Consistency, we have $\not (\mu R \pi)$, equivalently, by definition of $d$, $d(\pi)>d(\mu)$.

Take any distinct $\pi, \pi'\in\Delta(S)$ with $d(\pi)=d(\pi')$. Let $I=[\pi, \pi']$. If $C(I)=\pi$, then by Revealed Consistency we have $\neg(\pi R \pi')$. Equivalently, by definition of $d$, we obtain $d(\pi')>d(\pi)$, which contradicts the assumption that $d(\pi)=d(\pi')$. Similarly, if $C(I)=\pi'$, then we obtain $d(\pi)>d(\pi')$, which contradicts the assumption that $d(\pi)=d(\pi')$. Finally, suppose $x=C(I)=\gamma\,\pi+(1-\gamma)\pi'$ for some $\gamma\in (0, 1)$. By Step 4, we have $C(I)=C([x, \pi])$ since $[x, \pi]\subseteq I$. Therefore, by Revealed Consistency, we obtain $\neg(x R \pi)$, equivalently, $d(\pi)> d(x)=d(\gamma\,\pi+(1-\gamma)\pi')$.\smallskip

\smallskip
\noindent\textbf{Step 10.} For any convex $I$, if $\mu\not\in I$, then $C(I)\in \partial(I)=I\setminus\text{int}(I)$, the boundary of $I$.\medskip

Since $\mu\not\in I$, $\mu\neq \mu_I=\alpha\, C(I)+(1-\alpha)\,\mu$. Hence, $\succsim_I\neq \succsim$. Take any $\pi\in\text{int}(I)$. By \nameref{EXTM}, $\succsim_I\neq\succsim_{\{\pi\}}$, i.e., $C(I)\neq C(\pi)=\pi$. Hence, we have $C(I)\neq \pi$ for any $\pi\in\text{int}(I)$. That is, $C(I)\in \partial(I)$.\smallskip

\medskip
\noindent\textbf{Step 11.} $d$ is locally nonsatiated with respect to $\mu$.\medskip
\medskip

Now take any $\pi\neq \mu$ and $\epsilon>0$. If $\mu\in B(\pi, \frac{\epsilon}{2})$, then by Step 9, $\mu$ is the desired alternative; i.e., $\|\mu-\pi\|<\epsilon$ and $d(\mu)<d(\pi)$. If $\mu\not\in B(\pi, \frac{\epsilon}{2})$, then by Step 10, $\pi^*=C(B(\pi, \frac{\epsilon}{2}))$ is different from $\pi$ and $\mu$. By Step 4, we also have $\pi^*=C([\pi, \pi^*])$. By the arguments in Step 9, $\pi^*=C([\pi, \pi^*])$ implies $d(\pi)>d(\pi^*)$. Hence, we found $\pi^*$ such that $d(\pi)>d(\pi^*)$ and $\|\pi-\pi^*\|<\epsilon$.\medskip

\bigskip
\noindent\textbf{Necessity.} Suppose the family of preference relations $\{\succsim_I\}_{I\in\mathscr{I}}$ admits the \nameref{wiuseu} representation with $(w, u, d_\mu, \succ_{d_\mu})$ where $d_\mu$ is a locally nonsatiated, continuous distance function. It is immediate that \nameref{SEU} are satisfied. We now show the necessity of the other axioms. When $w=1$, the axioms are trivially satisfied since $\succsim=\succsim_I$ for all $I$. Suppose now $w<1$.

For simplicity, let us define a choice function $C$ on $\mathscr{I}$ such that 
\[C(I)=\min(\arg\min d_\mu(I), \succ_{d_\mu})\text{ for each }I.\]
Note that for any $I, I'$, $\succsim_I=\succsim_{I'}$ iff $C(I)=C(I')$.

\medskip
\noindent \nameref{CON}. Take any $I\in\mathscr{I}$. By the representation, there is $\pi\in I$ such that $\pi=C(I)$. Hence, we also have $C(\{\pi\})=C(I)$, which implies that $\succsim_{I}=\succsim_{\{\pi\}}$.

\medskip
\noindent \nameref{IB}. Take any $I_1, I_2, I_3\in\mathscr{I}$ such that $I_3\subseteq I_2\subseteq I_1$ and $\succsim_{I_1}=\succsim_{I_3}$. Note that $\succsim_{I_1}=\succsim_{I_3}$ implies that $C(I_1)=\min \big(\arg\min d_\mu(I_1), \succ_{d_\mu}\big)=\min \big(\arg\min d_\mu(I_3), \succ_{d_\mu}\big)=C(I_3)$. Hence, $C(I_1)\in I_3$. Since $I_3\subseteq I_2$, $C(I_1)\in I_2$. Therefore, $C(I_1)=\min\big(\arg\min d_\mu(I_2), \succ_{d_\mu}\big)$ since $C(I_1)=\min\big(\arg\min d_\mu(I_1), \succ_{d_\mu})$ and $C(I_1)\in I_2\subseteq I_1$. Hence, $\succsim_{I_1}=\succsim_{I_2}$.

\medskip
\noindent\nameref{EXTM}. Take any convex $I, I'\in\mathscr{I}$ with $I'\subseteq \text{int}(I)$ and $\succsim_I\neq \succsim$. Note that $\succsim_I\neq \succsim$ is equivalent to $\mu\not\in I$. We need to show that $C(I)\neq C(I')$, which is equivalent to $C(I)\not\in I'$ since $I'\subset I$. Hence, we shall show that $C(I)\in I\setminus I'$. Since $I'\subseteq \text{int}(I)$, it is sufficient to show that $C(I)\in\partial(I)$. 

By way of contradiction, suppose $C(I)\in \text{int}(I)$. Then there is $\epsilon>0$ such that $B(C(I), \epsilon)\subseteq I$. By local nonsatiation of $d_\mu$ and $\mu\not\in I$, there is $\pi\in B(C(I), \epsilon)$ such that $d_\mu(\pi)<d_\mu(C(I))$, which contradicts the fact that $C(I)=\arg\min d_\mu(I)$.

\medskip
\noindent\nameref{RCONT}. Since $d$ is continuous, it is enough to prove that $\pi' R \pi$ if and only if $d(\pi')\ge d(\pi)$ for any distinct $\pi, \pi'$.

Take any distinct $\pi, \pi'$. If $d(\pi)>d(\pi')$, by continuity of $d$, there is $\epsilon>0$ such that $d(\pi_\epsilon)>d(\pi'_\epsilon)$ for any $\pi_\epsilon\in B(\pi, \epsilon)$ and $\pi'_\epsilon\in B(\pi', \epsilon)$. Hence, $\succsim_{\{\pi_\epsilon, \pi'_{\epsilon}\}}=\succsim_{\{\pi'_{\epsilon}\}}$ for any $\pi_\epsilon\in B(\pi, \epsilon)$ and $\pi'_\epsilon\in B(\pi', \epsilon)$. Hence, $\neg(\pi' R \pi)$. In other words, $\pi' R \pi$ implies $d(\pi')\ge d(\pi)$.

Suppose now $d(\pi')\ge d(\pi)$. We shall prove that $\pi' R \pi$. If $d(\pi')>d(\pi)$, then we have $\succsim_{\{\pi, \pi'\}}=\succsim_{\{\pi\}}$, which implies that $\pi' R \pi$.

Suppose now $d(\pi')=d(\pi)$. By way of contradiction, suppose that $\neg(\pi' R \pi)$, i.e., there is $\epsilon>0$ such that $\succsim_{\{\pi_\epsilon, \pi'_{\epsilon}\}}=\succsim_{\{\pi'_{\epsilon}\}}$ for any $\pi_\epsilon\in B(\pi, \epsilon)$ and $\pi'_\epsilon\in B(\pi', \epsilon)$. Hence,  $d(\pi_{\epsilon})\ge d(\pi'_{\epsilon})$ and $\pi_{\epsilon}\succ_{d_\mu} \pi'_{\epsilon}$ for any $\pi_\epsilon\in B(\pi, \epsilon)$ and $\pi'_\epsilon\in B(\pi', \epsilon)$.  

Note that $\pi\neq \mu$. Hence, by local nonsatiation, there exists $\pi_\epsilon\in B(\pi, \epsilon)$ such that $d(\pi_\epsilon)<d(\pi)$. Then we have $d(\pi')= d(\pi)>d(\pi_\epsilon)$, which contradicts the assumption that $d(\pi_{\epsilon})\ge d(\pi'_{\epsilon})$ and $\pi_{\epsilon}\succ_{d_\mu} \pi'_{\epsilon}$ for any $\pi_\epsilon\in B(\pi, \epsilon)$ and $\pi'_\epsilon\in B(\pi', \epsilon)$.

\medskip
\noindent\nameref{RCON}. Take any distinct $\pi, \pi'\in\Delta$ such that $\succsim_{[\pi, \pi']}=\succsim_{\{\pi\}}$. We shall prove that $\neg(\pi\,R \pi')$. By way of contradiction, suppose $\pi\, R \pi'$, which implies $d(\pi)\ge d(\pi')$. By the representation, $\succsim_{[\pi, \pi']}=\succsim_{\{\pi\}}$ implies that $d(\pi)\le d(\alpha \pi+(1-\alpha)\pi')$ for any $\alpha\in [0, 1]$.

If $d(\pi)> d(\pi')$, then we cannot have $\succsim_{[\pi, \pi']}=\succsim_{\{\pi\}}$. If $d(\pi)= d(\pi')$, then there is $\alpha\in (0, 1)$ such that $d(\pi)>d(\alpha \pi+(1-\alpha)\pi')$, which also contradicts $\succsim_{[\pi, \pi']}=\succsim_{\{\pi\}}$.

\medskip
\noindent\nameref{Lin}. Take any $p, q, r, r'\in\Delta(X)$, $\pi, \pi'\in\Delta(S)$, and $\lambda\in[0, 1]$ such that $p\,E\,q\sim_{\{\pi\}}\,r$ and $p\,E\,q\sim_{\{\pi'\}}\,r'$. If $u(p)=u(q)$, it is trivial to obtain the desired result. Suppose $u(p)\neq u(q)$. We then have $\mu_\pi(E)=\frac{u(r)-u(q)}{u(p)-u(q)}$ and $\mu_{\pi'}(E)=\frac{u(r')-u(q)}{u(p)-u(q)}$. Since $\mu_\pi$ is linear in $C(\pi)=\pi$, by linearly mixing the above two equations, we obtain
 \[\mu_{\lambda\,\pi+(1-\lambda)\,\pi'}(E)=\lambda\,\mu_\pi(E)+(1-\lambda)\,\mu_{\pi'}(E)=\frac{\lambda\,u(r)+(1-\lambda)\,u(r')-u(q)}{u(p)-u(q)}.\]
Hence, we obtain 
$p\,E\,q\sim_{\{\lambda\,\pi+(1-\lambda)\,\pi'\}}\,\lambda\, r+(1-\lambda)\,r'$.

\medskip
\noindent\nameref{EInd}. Take any $E\subset S$, and $\pi, \pi'\in \Delta(S)$ with $\pi(E)\geq\pi'(E)$. Note that we then have $\mu_\pi(E)\geq\mu_{\pi'}(E)$. Hence, for any $p, q, r\in\Delta(X)$, $p\succsim q$, we have $p\, E\, q\succsim_{\{\pi'\}} r$ iff $\mu_{\pi'}(E)u(p)+ (1-\mu_{\pi'}(E))u(q)\geq u(r)$, which implies $\mu_\pi(E)u(p)+ (1-\mu_\pi(E))u(q)\geq u(r)$ hence $p\, E\, q\succsim_{\{\pi\}} r$.

\medskip
\noindent\nameref{IComp}. Take $p, q\in\Delta(X)$, $E\subset S$, and $\pi\in \Delta(S)$. If $\succsim=\succsim_{\{\pi\}}$ $\mu_\pi=w \mu + (1-w)\pi = \mu$, and thus $\pi=\mu$ and so $p\, E\, q\sim \pi(E)p+(1-\pi(E))\, q$ for all $E$.

\subsection{Proof of Theorem \ref{repthrm}}

\noindent\textbf{Sufficiency.} We first prove sufficiency.  Consider a family of preferences $\{\succsim_I\}_{I\in\mathscr{I}}$ that satisfies \nameref{SEU}, \nameref{COMP}, \nameref{CON},  \nameref{IB}, \nameref{EXTM}, \nameref{RCONT}, and \nameref{RCON}.

\bigskip
\noindent\textbf{Step 1.} Under \nameref{SEU}, there is $(u, (\mu_{I})_{I\in\mathscr{I}})$ such that for each $I\in\mathscr{I}$, the conditional preference relation $\succsim_I$  admits the SEU representation with $(u, \mu_I)$. As usual, $u$ is unique up to a positive affine transformation, and $\mu_I$ is unique. Let $\mu=\mu_{\Delta(S)}$.

\bigskip
\noindent\textbf{Step 2.} ($\mu_{\{\pi\}}=\pi$). Take a probability distribution $\pi\in\Delta(S)$ and let $I=\{\pi\}$ be a singleton set. By \nameref{COMP} and Step 1, for any $p, q\in \Delta(X)$ and $E\subset S$,
\begin{equation}
p\, E\, q\sim_{\{\pi\}} \pi(E)p+(1-\pi(E))q  \text{ iff }  \mu_{\pi}(E)u(p)+(1-\mu_{\pi}(E))u(q)=\pi(E)u(p)+(1-\pi(E))u(q).
\end{equation}
By setting $u(p)\neq u(q)$, we obtain $\mu_{\pi}(E)=\pi(E)$ for each $E$. Therefore, we have $\mu_{\{\pi\}}=\pi$.

By the proof of Theorem \ref{repthrm2} and Step 2, it is now sufficient to prove that \nameref{COMP} implies \nameref{Lin}, \nameref{EInd}, and \nameref{IComp}.

\bigskip
\noindent\textbf{Step 3.} \nameref{COMP} implies \nameref{Lin}.\medskip

Take any $p, q, r, r'\in\Delta(X)$, $\pi, \pi'\in\Delta(S)$, and $\lambda\in[0, 1]$ such that $p\,E\,q\sim_{\{\pi\}}\,r$ and $p\,E\,q\sim_{\{\pi'\}}\,r'$. If $u(p)=u(q)$, then we obtain the desired result trivially. Suppose now that $u(p)\neq u(q)$. We then have $\mu_\pi(E)=\pi(E)=\frac{u(r)-u(q)}{u(p)-u(q)}$ and $\mu_{\pi'}(E)=\pi'(E)=\frac{u(r')-u(q)}{u(p)-u(q)}$. By linearly mixing the above two equations, we obtain
 \[\mu_{\lambda\,\pi+(1-\lambda)\,\pi'}(E)=\lambda\,\pi(E)+(1-\lambda)\,\pi'(E)=\frac{\lambda\,u(r)+(1-\lambda)\,u(r')-u(q)}{u(p)-u(q)}.\]
Hence, we obtain 
$p\,E\,q\sim_{\{\lambda\,\pi+(1-\lambda)\,\pi'\}}\,\lambda\, r+(1-\lambda)\,r'$.

\bigskip
\noindent\textbf{Step 4.} \nameref{COMP} implies \nameref{EInd}. \medskip

Take any $E\subset S$, and $\pi, \pi'\in \Delta(S)$ with $\pi(E)\geq\pi'(E)$. Since $\mu_\pi=\pi$ and $\mu_{\pi'}=\pi'$, $\pi(E)\geq\pi'(E)$ implies $\mu_\pi(E)\geq\mu_{\pi'}(E)$. Hence, for any $p, q, r\in\Delta(X)$, $p\succsim q$, we have $p\, E\, q\succsim_{\{\pi'\}} r$ iff $\mu_{\pi'}(E)u(p)+ (1-\mu_{\pi'}(E))u(q)\geq u(r)$, which implies $\mu_\pi(E)u(p)+ (1-\mu_\pi(E))u(q)\geq u(r)$ hence $p\, E\, q\succsim_{\{\pi\}} r$.

\bigskip
\noindent\textbf{Step 5.} \nameref{COMP} implies \nameref{IComp}.\medskip

Take any $\pi$ and $p, q\in\Delta(X)$, $E\subset S$. Since $\mu_\pi=\pi$, $\succsim=\succsim_{\{\pi\}}$ implies $\mu=\pi$. Note that the expected utility of $p\,E\,q$ is $\mu(E) u(p)+(1-\mu(E)) u(q)$, which is also equivalent to the expected utility of $\pi(E)\,p+(1-\pi(E))\,q$ since $\mu=\pi$. Therefore, $p\, E\, q\sim \pi(E)\,p+(1-\pi(E))\, q$.

\bigskip
\noindent\textbf{Necessity.} By the proof of Theorem \ref{repthrm2}, it is enough to show the necessity of \nameref{COMP}. Take any $p, q\in \Delta(X)$, $\pi\in\Delta(S)$, and $E\subset S$. Since $\mu_{\{\pi\}}=\pi$, the expected utility of $p\,E\,q$ is $\pi(E) u(p)+(1-\pi(E)) u(q)$, which is also the expected utility of $\pi(E)\,p+(1-\pi(E))\,q$. Therefore, $p\, E\, q\sim_{\{\pi\}} \pi(E)\,p+(1-\pi(E))\, q$.

\subsection{Proof of Proposition \autoref{uniqueness}}

The uniqueness of $u$ and $\mu_I$ are standard. Since $\mu=\mu_{\Delta(S)}$, $\mu$ is also unique. To prove (iii) or $\succ_{d_\mu}=\succ'_{d'_\mu}$, take arbitrary distinct $\pi, \pi'$ with $\pi\succ_{d_\mu} \pi'$. First, $\pi\succ_{d_\mu} \pi'$ implies that $d_\mu(\pi')\le d_{\mu}(\pi)$. Hence, we have $\pi'=\min(\arg\min d_\mu(\{\pi, \pi'\}, \succ_{d_\mu})$, consequently, $\mu_{\{\pi, \pi'\}}=\pi'$. Then $\mu_{\{\pi, \pi'\}}=\pi'$ implies that $\pi'=\min(\arg\min d'_\mu(\{\pi, \pi'\}, \succ'_{d'_\mu})$, which implies that $d'_\mu(\pi')\le d'_{\mu}(\pi)$. Hence, $\pi\succ'_{d'_\mu} \pi'$. 

To prove (iv), it is enough to show that $d_\mu(\pi)\ge  d_\mu(\pi')$ if and only if $d'_\mu(\pi)\ge  d'_\mu(\pi')$ for any $\pi, \pi'$. By way of contradiction, suppose there are $\pi, \pi'$ such that $d_\mu(\pi)>d_\mu(\pi')$ and $d'_\mu(\pi)\le  d'_\mu(\pi')$.

The definition of the distance function implies that $\pi\neq \mu$. Then by local nonsatiation of $d'$, for any $\epsilon>0$, there is $\pi_\epsilon\in B(\pi, \epsilon)$ such that $d'_\mu(\pi_\epsilon)<d'_\mu(\pi)\le  d'_\mu(\pi')$. However, since $d$ is continuous, there is a small enough $\epsilon>0$ such that $d_\mu(\pi)>d_\mu(\pi')$ implies $d_\mu(\pi_\epsilon)>d_\mu(\pi')$. Hence, we obtained $\succsim_{\{\pi_{\epsilon}, \pi'\}}\neq \succsim'_{\{\pi_{\epsilon}, \pi'\}}$ when $\epsilon$ is small enough, a contradiction.

\subsection{Proof of Proposition \ref{DCprop}}

Take any $E \subset S$, $\beta \in (0,1]$ and $f,p,h \in \cal F$ such that $f Eh\sim pEh$ and $f Eh\sim_{I^\beta_E} pEh$. By the \nameref{iuseu} representation:
\[f Eh\sim pEh  \text{ if and only if } \sum_{s\in E} \mu(s)u(f(s))=\mu(E)u(p)\]  
and
\[f Eh\sim_{I^\beta_E} pEh \text{ if and only if }\sum_{s\in E} \mu_{I^\beta_E}(s)u(f(s))=\mu_{I^\beta_E}(E)u(p).\]
By \nameref{IDC}, we thus have 
\begin{equation*}
\frac{1}{\mu(E)} \sum_{s\in E} \mu(s)u(f(s))=\frac{1}{\mu_{I^\beta_E}(E)} \sum_{s\in E} \mu_{I^\beta_E}(s)u(f(s)),
\end{equation*}
which is equivalent to $\frac{\mu(s)}{\mu(E)}=\frac{\mu_{I^\beta_E}(s)}{\mu_{I^\beta_E}(E)}=\frac{\mu_{I^\beta_E}(s)}{\beta}$
for each $s \in E$.

\subsection{Proof of Theorem \autoref{bayesdistprop}}

Without loss of generality, suppose $\mu_1=\min_{i} \mu_i$. Take any $\beta\in (0, 1]$ and $i$. Let $I=I^\beta_E$ where $E=\{s_1, s_i\}$. Consider the Inertial Updating with $I$:
\[\min_{\pi} d_1(\pi_1)+d_i(\pi_i)\text{ s.t. }\pi_1+\pi_i=\beta.\]
The first order condition gives $d'_i(\pi_i)=d'_1(\pi_1)$. The second-order condition is satisfied since each $d_i$ is strictly convex. Since Informational Dynamic Consistency is satisfied, we have $\pi_i=\beta\,\frac{\mu_i}{\mu_i+\mu_1}$. Then $d'_i(\frac{\mu_i}{\mu_i+\mu_1}\,\beta)=d'_1(\frac{\mu_1}{\mu_1+\mu_i}\,\beta)$. Equivalently, \[d'_i(\pi)=d'_1(\frac{\mu_1}{\mu_i}\,\pi)\text{ for any }\pi\in [0, \frac{\mu_i}{\mu_i+\mu_1}].\]
Note that 
\[d_i(\pi)-d_i(\mu_i)=\int^\pi_{\mu_i} d'_i(\tilde{\pi}) d\tilde{\pi}=\int^\pi_{\mu_i}d'_1(\frac{\mu_1}{\mu_i}\tilde{\pi}) d\tilde{\pi}=\frac{\mu_i}{\mu_1}\,\int^{\frac{\mu_1}{\mu_i}\pi}_{\mu_1}d'_1(\bar{\pi}) d\bar{\pi}=\frac{\mu_i}{\mu_1}\big(d_1(\frac{\mu_1}{\mu_i}\pi)-d_1(\mu_1)\big).\] Therefore, $d_i(\pi)=\frac{\mu_i}{\mu_1}\,d_1(\frac{\mu_1}{\mu_i}\pi)+c_i$ where $c_i=d_i(\mu_i)-\frac{\mu_i}{\mu_1}\,d_1(\mu_1)$. Then, the distance function is given by
\[d_\mu(\pi)=\sum^n_{i=1} d_i(\pi_i)=\sum^n_{i=1} \frac{\mu_i}{\mu_1}\,d_1(\frac{\mu_1}{\mu_i}\pi_i)+\sum^n_{i=1}\,c_i\text{ where }\pi_i\in [0, \frac{\mu_i}{\mu_i+\mu_1}].\]
Let $f(t)=\frac{d_1(\mu_1\,t)}{\mu_1}+\sum^n_{i=1} c_i$. Then we have $d_\mu(\pi)=\sum^n_{i=1} \mu_i\,f(\frac{\pi_i}{\mu_i})$ whenever $\pi_i\le \frac{\mu_i}{\mu_i+\mu_1}$.

\subsection{Proof of Proposition \ref{prop:interval}}

Take any \nameref{iuseu} representation with a Bayesian divergence function $d_\mu(\pi)=-\sum_i \mu_i\,\sigma(\frac{\pi_i}{\mu_i})$. Let $I=\{\pi\mid \underline{\beta}\le \pi(E)\le \overline{\beta}\}$ and $\gamma=\mu_I(E)$. Since $\mu_{I^\gamma_E}\in I^\gamma_E=I\cap I^\gamma_E$ and $\gamma=\mu_I(E)$, we should have $\mu_I=\mu_{I^\gamma_E}$. Moreover, $\gamma\in [\underline{\beta}, \overline{\beta}]$ minimizes
\[d_\mu(\mu_{I^{x}_E})=-\mu(E)\, \sigma\big(\frac{x}{\mu(E)}\big)-(1-\mu(E))\, \sigma\big(\frac{1-x}{1-\mu(E)}\big),\]
    which is decreasing when $x\le \mu(E)$ and is increasing when $x\ge \mu(E)$. Hence, $\gamma=\arg\min_{x\in [\underline{\beta}, \overline{\beta}]}|x-\mu(E)|$.

\subsection{Proof of Theorem \ref{prop:renyi}}

\noindent\textbf{Sufficiency.} Suppose that \nameref{SI} is satisfied for every full-support
prior $\mu$. Let $E=\{s_1, s_2, s_3\}$. Take any $\beta\in (0, \frac{1}{4})$ and $\alpha\in (2\beta, 4\beta)$ and $\lambda\in (0, 1]$. Let $I=\{\pi\in \Delta(S)|\pi(s_1)+\pi(s_2)+\pi(s_3)=\alpha\text{ and }\frac{\pi(s_1)}{2}+\pi(s_2)=\beta\}$ and $I'=\{\pi\in \Delta(S)|\pi(s_1)+\pi(s_2)+\pi(s_3)=\lambda\,\alpha\text{ and }\frac{\pi(s_1)}{2}+\pi(s_2)=\lambda\,\beta\}$. Note that $I$ and $I'$ are proportional on $E$. By Scale Invariance, we have \[p\,\{s\}\,r\sim_{I} q\,\{s'\}\,r\text{ iff } p\,\{s\}\,r\sim_{I'} q\,\{s'\}\,r\text{ for any }s, s'\in E,\]
equivalently, 
\[\frac{\mu_I(s)}{\mu_I(s')}=\frac{u(q)-u(r)}{u(p)-u(r)}=\frac{\mu_{I'}(s)}{\mu_{I'}(s')}.\]
Hence, 
\[\big(\mu_{I'}(s_1), \mu_{I'}(s_2), \mu_{I'}(s_3)\big)=\lambda\,\big(\mu_I(s_1), \mu_I(s_2), \mu_I(s_3)\big).\]
Note that the vector $\big(\mu_I(s_1), \mu_I(s_2), \mu_I(s_3)\big)$ is the unique solution to the following concave optimization problem:
\[\max_{x\in [0, 2\,\beta]} \mu(s_1)\sigma\big(\frac{x}{\mu(s_1)}\big)+\mu(s_2)\sigma\big(\frac{\beta-x/2}{\mu(s_2)}\big)+\mu(s_3)\sigma\big(\frac{\alpha-\beta-x/2}{\mu(s_3)}\big).\]
The first-order condition gives
\[\frac{\sigma'\big(\frac{\beta-x/2}{\mu(s_2)}\big)+\sigma'\big(\frac{\alpha-\beta-x/2}{\mu(s_3)}\big)}{2}=\sigma'\big(\frac{x}{\mu(s_1)}\big).\]
Let $\mu(s_1)=\mu(s_3)$. Since $\sigma'$ is strictly decreasing and continuous on $(0, \infty)$, and $1>2\beta> \alpha-2\beta>0$, the first-order condition always has an interior solution. Similarly for $I'$, we obtain
\[\frac{\sigma'\big(\frac{\lambda\,\beta-\lambda\,x/2}{\mu(s_2)}\big)+\sigma'\big(\frac{\lambda\,\alpha-\lambda\,\beta-\lambda\,x/2}{\mu(s_3)}\big)}{2}=\sigma'\big(\frac{\lambda\,x}{\mu(s_1)}\big).\]

For any $a, b>0$, by appropriately choosing $\alpha, \beta, x$ and $\mu(s_2), \mu(s_3)=\mu(s_1)$, we can obtain $\sigma'(a)/2+\sigma'(b)/2=\sigma'\big(\frac{x}{\mu(s_1)}\big)$ where $a=\frac{\beta-x/2}{\mu(s_2)}$ and $b=\frac{\alpha-\beta-x/2}{\mu(s_3)}.$ Hence, we obtain
\[f(\frac{\sigma'(\lambda\,a)+\sigma'(\lambda\,b)}{2})=\frac{\lambda\,x}{\mu(s_1)}=\lambda\,f(\frac{\sigma'(a)+\sigma'(b)}{2})\text{ for any }a, b>0,\]
where $f=\sigma'^{-1}$. Let $h_\lambda(t)=\sigma'(\lambda\,f(t))$ and $x'=\sigma'(a)$ and $y'=\sigma'(b)$. Then we obtain 
\[\frac{h_\lambda(x')+h_\lambda(y')}{2}=h_\lambda(\frac{x'
+y'}{2}).\]
Given that $\sigma'$ is a monotonic function (i.e., $h_\lambda$ is monotonic), the above gives a typical Cauchy functional equation where the solution is $h_\lambda\,(x)=g(\lambda)\,x+\delta(\lambda)$ for some $g, \delta$. Note that $h_\lambda\,(x)=g(\lambda)\,x+\delta(\lambda)$ is equivalent to $\sigma'(\lambda\,a)=g(\lambda)\,\sigma'(a)+\delta(\lambda)$ by setting $\sigma'(a)=x$. By setting $a=1$, we obtain $\sigma'(\lambda)=g(\lambda)\sigma'(1)+\delta(\lambda)$. Let $\sigma'(1)=c$. Then we have $\sigma'(\lambda)=c\,g(\lambda)+\delta(\lambda)$. Hence, $\delta(\lambda)=\sigma'(\lambda)-c\,g(\lambda)$. Plugging the equation for $\delta$ back into the original equation, we obtain $\sigma'(\lambda\,a)=\sigma'(\lambda)+g(\lambda)(\sigma'(a)-c)$. Let $F(a)=\sigma'(a)-c$. Then we have $F(\lambda\,a)=g(\lambda)\,F(a)+F(\lambda)$. 

If $g=1$, then we obtain $F(\lambda\,a)=F(a)+F(\lambda)$, which is the multiplicative Cauchy functional equation. Hence, $F(a)=\zeta\,\log(a)$. Hence, $\sigma'(a)=\zeta\,\log(a)+c$ with $\zeta<0$, equivalently, $\sigma(a)=\zeta\, a\,\log(a)+(c-\zeta)\,a+d$. Hence, $d^\sigma_\mu$ is the $\alpha$-divergence with $\alpha=1$.

Assume now that $g\neq 1$. Let us focus on the case $a\le 1$. Interchanging $a$ and $\lambda$, $F(\lambda\, a)=g(\lambda)\,F(a)+F(\lambda)=g(a)\,F(\lambda)+F(a)$. Hence,
$(g(a)-1)F(\lambda)=(g(\lambda)-1)F(a)$. Hence, there is $k\neq 0$ such that $g(\lambda)=k\,F(\lambda)+1$. Hence, $F(\lambda\,a)=g(\lambda)\,F(a)+F(\lambda)=k\,F(a)F(\lambda)+F(a)+F(\lambda)$; equivalently, $k\,F(\lambda a)+1=(k\,F(a)+1)(k\,F(\lambda)+1)$, equivalently, $g(\lambda\, a)=g(\lambda)g(a)$. Since this is the multiplicative Cauchy functional equation and $g$ is monotonic, we have $g(\lambda)=\lambda^\rho$, consequently, $F(\lambda)=\frac{\lambda^\rho-1}{k}$ for any $\lambda\le 1$.  

Let us consider the case $a>1$. We can find a small enough $\lambda$ such that $\lambda\, a<1$. Then $\frac{(\lambda\,a)^\rho-1}{k}=F(\lambda\,a)=g(\lambda)\,F(a)+F(\lambda)=\lambda^\rho\,F(a)+\frac{\lambda^\rho-1}{k}$. Hence, we obtain $F(a)=\frac{a^\rho-1}{k}$ for any $a>0$. Consequently, $\sigma'(a)=A\,a^\rho+B$ for some $A, B$. Hence, for some $A', B', C'$, $\sigma(x)=A'\,x^{\rho+1}+B'x+C'$ when $\rho\neq -1$ and $\sigma(x)=A'\,\log(x)+B'x+C'$ when $\rho=-1$. In either case, $d^\sigma_\mu$ is an $\alpha$-divergence.

\medskip
\noindent\textbf{Necessity.} Take any $\alpha$-divergence and convex information constraints $I$ and $I'$ that are proportional on $E$. By the first part of the proof of Theorem \ref{prop:renyi}, we shall prove that $(\mu_{I'}(s))_{s\in E}=\lambda\,(\mu_{I}(s))_{s\in E}$ for some $\lambda>0$.

Since $I=I(E)\times I(E^c)\subset \Delta(S)$, we have $\pi(E)+\pi'(E^c)=1$ for any $(\pi(s))_{s\in E}\in I(E)$ and $(\pi'(s))_{s\in E^c}\in I(E^c)$. Hence, there is $\alpha\ge 0$ such that $I\subseteq I^\alpha_E$. Similarly, there is $\beta\ge 0$ such that $I'\subseteq I^\beta_E$. If $\alpha=\beta=0$, then $(\mu_I(s))_{s\in E}=(\mu_{I'}(s))_{s\in E}=\textbf{0}$, which gives us the desired result. Suppose $\alpha, \beta>0$. Since $I$ and $I'$ are proportional on $E$, $\frac{I(E)}{\alpha}=\frac{I'(E)}{\beta}$. 

Since $I=I(E)\times I(E^c)$, $(\mu_I(s))_{s\in E}$ is the unique solution to
\[\max \sum_{s\in E}\mu(s)\sigma(\frac{\pi(s)}{\mu(s)})\text{ subject to }(\pi(s))_{E}\in I(E).\]
Equivalently, since $\sigma$ is a power function, $(\frac{\mu_I(s)}{\alpha})_{s\in E}$ is the unique solution to
\[\max \sum_{s\in E}\mu(s)\sigma(\frac{\pi(s)}{\mu(s)})\text{ subject to }(\pi(s))_{E}\in \frac{I(E)}{\alpha}.\footnote{The same argument is true if $\sigma(x)=\log(x)$ or $\sigma(x)=-x\log (x)$.}\]
Similarly, $(\frac{\mu_{I'}(s)}{\beta})_{s\in E}$ is the unique solution to
\[\max \sum_{s\in E}\mu(s)\sigma(\frac{\pi(s)}{\mu(s)})\text{ subject to }(\pi(s))_{E}\in \frac{I'(E)}{\beta}.\]
Since $\frac{I(E)}{\alpha}=\frac{I'(E)}{\beta}$, $(\frac{\mu_I(s)}{\alpha})_{s\in E}=(\frac{\mu_{I'}(s)}{\beta})_{s\in E}$.

\subsection{Proof of Proposition \ref{prop:kl}}

\noindent\textbf{(i) implies (iii)}. Take any disjoint events $A$ and $B$ with $\mu(A)<\mu(B)$, and let $I=\{\pi\in\Delta(S)\mid \pi(A)\ge \pi(B)\}$. Consider $p,q, f$ so that $p A\cup B f \sim q$. The representation and \nameref{IIE} implies $\mu(A\cup B)u(p) + \sum_{s \in (A\cup B)^c}u(f(s))\mu(s)=u(q)=\mu_I(A\cup B)u(p) + \sum_{s \in (A\cup B)^c}u(f(s))\mu_I(s)$, hence, $\mu(A\cup B)=\mu_I(A\cup B)$. Let $\mu(A)=a$ and $\mu(B)=b$ with $a<b$, and $\mu_I(A)=\alpha$ and $\mu_I(B)=\beta$ with $\alpha\ge \beta$.  Then we have
\[d_\mu(\mu_I)=-a\,\sigma(\frac{\alpha}{a})-b\,\sigma(\frac{\beta}{b})-(1-a-b)\,\sigma(\frac{1-\alpha-\beta}{1-a-b}).\]
Hence, $\mu_I$ solves
\[\max_{y, z\ge 0} a\,\sigma(\frac{y+z}{a})+b\,\sigma(\frac{y}{b})+(1-a-b)\,\sigma(\frac{1-2y-z}{1-a-b}).\]
If the above problem has inner solutions, then the FOCs for $z$ and $y$ should give
\[\sigma'(\frac{y+z}{a})=\sigma'(\frac{1-2y-z}{1-a-b})\]
and 
\[\sigma'(\frac{y+z}{a})+\sigma'(\frac{y}{b})=2\,\sigma'(\frac{1-2y-z}{1-a-b}),\]
respectively. Consequently, we have
\[\frac{y+z}{a}=\frac{y}{b}=\frac{1-2y-z}{1-a-b},\]
which cannot hold since $\frac{y+z}{a}>\frac{y}{b}$. Hence, given that $\sigma$ is strictly concave, we have $z=0$. Then, by the FOC for $y$,
\[\sigma'(\frac{y}{a})+\sigma'(\frac{y}{b})=2\,\sigma'(\frac{1-2y}{1-a-b}).\]

By Theorem 4, we can assume $d_\mu$ is $\alpha$-divergence, that is, either $\sigma'(x)=\beta\,x^{\alpha-1}+\gamma$ (when $\alpha\neq 1$) or $\sigma'(x)=a\,\log(x)+b$ (when $\alpha=1$). By \nameref{IIE}, $y=\frac{a+b}{2}$ should solve the above FOC. Hence, when $\sigma'(x)=\beta\,\log(x)+\gamma$, the FOC gives
\[\sigma'(\frac{y}{a})+\sigma'(\frac{y}{b})=\beta\,\log(\frac{y}{a})+\beta\log(\frac{y}{b})+2\gamma=\beta\log(\frac{y^2}{ab})+2\,\gamma=2\gamma=2\,\sigma'(\frac{1-2y}{1-a-b}).\]
Hence, we obtain $\sqrt{ab}=y=(\frac{a+b}{2})$, a contradiction.  

However, when $\sigma'(x)=\beta\,x^{\alpha-1}+\gamma$, the FOC gives
\[\sigma'(\frac{y}{a})+\sigma'(\frac{y}{b})=\beta\,(\frac{a}{y})^{1-\alpha}+\beta(\frac{b}{y})^{1-\alpha}+2\gamma=\beta\frac{a^{1-\alpha}+b^{1-\alpha}}{y^{1-\alpha}}+2\gamma=2\beta+2\gamma=2\,\sigma'(\frac{1-2y}{1-a-b});\]
equivalently, $\frac{a^{1-\alpha}+b^{1-\alpha}}{2}=y^{1-\alpha}=(\frac{a+b}{2})^{1-\alpha}$. Hence, we obtain $\alpha=0$; i.e., $d_\mu$ is a KL-divergence.

\medskip
\noindent\textbf{(ii) implies (iii)}. Take any disjoint events $A$ and $B$ with $\mu(A)<\rho\,\mu(B)$, and let $I=\{\pi\in\Delta(S)\mid \pi(A)\ge \rho\,\pi(B)\}$. Similar to the preceding step, \nameref{SIIE} implies that beliefs on $(A\cup B)^c$ are not affected by the information constraint $I$. Hence, $\mu(A\cup B)=\mu_I(A\cup B)$. Let $\mu(A)=a$ and $\mu(B)=b$ with $a<\rho\,b$. By repeating the same arguments, we obtain the following FOC: 
\[\rho\,\sigma'(\frac{\rho\,y}{a})+\sigma'(\frac{y}{b})=(1+\rho)\,\sigma'(\frac{1-(1+\rho)y}{1-a-b}).\]
By \nameref{SIIE}, $y=\frac{a+b}{1+\rho}$ should solve the above FOC. Hence, we obtain 
\[\rho\,\sigma'(\frac{\frac{\rho}{1+\rho}(a+b)}{a})+\sigma'(\frac{\frac{1}{1+\rho}(a+b)}{b})=(1+\rho)\,\sigma'(1).\]
Let $x=\frac{\frac{\rho}{1+\rho}(a+b)}{a}$ and $y=\frac{\frac{1}{1+\rho}(a+b)}{b}$ and $f(z)=z\,\sigma'(z)$.\footnote{Note that any $x>1$ and $y\in (0, 1)$ can be attained by appropriately choosing $\rho, a, b$.} Then
$\frac{1-y}{x-y}f(x)+\frac{x-1}{x-y}f(y)=f(1),$ equivalently, $(y-1)(f(x)-f(1))=(x-1)(f(y)-f(1))$. Hence, $f(x)-f(1)=k(x-1)$ for some $k$, equivalently, $\sigma'(x)=\frac{k(x-1)+\delta}{x}$ for some $\delta$; equivalently, $\sigma(x)=(\delta-k)\,\log(x)+k\,x+c$. Hence, $d_\mu$ is a KL-divergence.

\medskip
\noindent\textbf{(iii) implies (ii)}. Suppose $d_\mu$ is a KL divergence. Take any disjoint events $A$ and $B$, and let $I=\{\pi\in\Delta(S)\mid \pi(A)\ge \rho\pi(B)\}$. If $\mu\in I$, then $\mu_I=\mu$, and the result follows immediately. Suppose $\mu(A)<\rho\mu(B)$. Let $\mu(A)=a$ and $\mu(B)=b$. Since the constraint binds, let $\mu_I(B)=y$, so that $\mu_I(A)=\rho y$. Then $y$ solves
\[
\max_y \;
a\log\left(\frac{\rho y}{a}\right)
+b\log\left(\frac{y}{b}\right)
+(1-a-b)\log\left(\frac{1-(1+\rho)y}{1-a-b}\right).
\]
The FOC gives
\[
\frac{a+b}{y}
=
\frac{(1+\rho)(1-a-b)}{1-(1+\rho)y}.
\]
Hence, $y=\frac{a+b}{1+\rho}$, and therefore
\[
\mu_I((A\cup B)^c)
=
1-(1+\rho)y
=
1-a-b
=
\mu((A\cup B)^c).
\]
Moreover, KL divergence preserves relative probabilities within $(A\cup B)^c$, and hence
$\mu_I(s)=\mu(s)$ for every $s\in(A\cup B)^c$. Therefore, Strong Independence of Irrelevant Events is satisfied.

\medskip
\noindent\textbf{(iii) implies (i)}. The necessity of \nameref{IIE} follows from the above step. The necessity of \nameref{SI} follows from Theorem \autoref{prop:renyi}. 

\subsection{Proof of Proposition \ref{NestBU}} 

Note that $\mu_{I^m_\sigma}$ is the unique solution to the following constrained optimization problem:
\[\max \sum_{s_i\in S}\mu_i\,\sigma\big(\frac{\pi_i}{\mu_i}\big)\]
\[\text{ subject to }\sum_{s_i\in S} \sigma'\big(\frac{P(m|s_i)}{P(m)}\big)\,\pi_i=\kappa^\sigma_m\text{ and } \sum^n_{i=1} \pi_i=1.\]
Then there are Lagrange multipliers $\lambda$ and $\delta$ such that the following FOC is satisfied:
\[\sigma'\big(\frac{\pi_i}{\mu_i}\big)=\lambda\,\sigma'\big(\frac{P(m|s)}{P(m)}\big)+\delta\text{ for each }i\le n.\]
Note that the FOC is satisfied when $\lambda=1$, $\delta=0$, and $\pi=\mu^{BU}(\cdot|m)=(\frac{P(m|s_i)}{P(m)}\mu_i)_{i\le n}$. Since the FOC is satisfied, $\pi=\mu^{BU}(\cdot|m)$ is a local maximizer. Since the objective function $\sum_{s_i\in S}\mu_i\,\sigma\big(\frac{\pi_i}{\mu_i}\big)$ is strictly concave and the constraints are linear, the above $\pi$ is the unique maximizer.

\subsection{Proof of Proposition \ref{prop:disagree}}

Take any two Bayesian divergences, $d_\mu=-\sum_{i\in S} \mu_i\,\sigma(\frac{\pi_i}{\mu_i})$ and $d'_\mu=-\sum_{i\in S} \mu_i\,\delta(\frac{\pi_i}{\mu_i})$. Suppose that $\mu_I=\mu'_I$ for any $I$, where $\mu_I$ and $\mu'_I$ are posteriors resulting from \nameref{iuseu} representations with $d_\mu$ and $d'_{\mu}$ respectively. 

When $I=\{\pi, \pi'\}$, the above assumption implies that $d_{\mu}(\pi)\le d_\mu(\pi')$ whenever $d'_{\mu}(\pi)< d'_{\mu}(\pi')$. Let us prove the following holds: $d_{\mu}(\pi)=d_{\mu}(\pi')$ iff $d'_\mu(\pi)=d'_{\mu}(\pi')$.

By way of contradiction, suppose $d_{\mu}(\pi)<d_{\mu}(\pi')$ and $d'_\mu(\pi)=d'_{\mu}(\pi')$. It is clear that 
$\mu\not\in I$. Take small enough $\epsilon>0$ such that $d_{\mu}(\pi)<d_\mu((1-\epsilon)\pi'+\epsilon\,\mu)$. However, since $d'_\mu$ is strictly convex, we have $d'_\mu(\pi)=d'_{\mu}(\pi')>(1-\epsilon)d'_{\mu}(\pi')+\epsilon\, d'_{\mu}(\mu)>d'_\mu((1-\epsilon)\pi'+\epsilon \mu)$. Hence, we have $\mu_{I'}\neq \mu'_{I'}$ where $I'=\{\pi, (1-\epsilon)\pi'+\epsilon \mu\}$. Therefore, $d_\mu$ and $d'_\mu$ are ordinally equivalent.

Let us first assume $|S|=3$ and $\mu=(a, b, c)\neq (\frac{1}{3}, \frac{1}{3}, \frac{1}{3})$. Without loss of generality, let us assume $a\neq b$. For any $r, s>0$ with $a\,r+b\,s<1$, let $U(r, s)=a\,\sigma(r)+b\,\sigma(s)+c\,\sigma(\frac{1-a\,r-b\,s}{c})$ and $V(r, s)=a\,\delta(r)+b\,\delta(s)+c\,\delta(\frac{1-a\,r-b\,s}{c})$. Since $U$ and $V$ are ordinally equivalent, their local level curves coincide, i.e., there is $\lambda=\lambda(r, s)$ such that $\nabla U=\lambda\,\nabla V$. Equivalently, $a\,(\sigma'(r)-\sigma'(t))=\lambda\,a\,(\delta'(r)-\delta'(t))$ and $b\,(\sigma'(s)-\sigma'(t))=\lambda\,b\,(\delta'(s)-\delta'(t))$ where $t=\frac{1-a\,r-b\,s}{c}.$ By cancelling out $\lambda, a, b$, we obtain $\frac{\sigma'(r)-\sigma'(t)}{\sigma'(s)-\sigma'(t)}=\frac{\delta'(r)-\delta'(t)}{\delta'(s)-\delta'(t)}$, equivalently, $\frac{\sigma'(r)-\sigma'(s)}{\delta'(r)-\delta'(s)}=\frac{\sigma'(s)-\sigma'(t)}{\delta'(s)-\delta'(t)}$. By setting $r=x+\epsilon, s=x$ and $t=y-\frac{a}{c}\epsilon$ where $y=\frac{1-(a+b)x}{c}$, we obtain 
$\frac{\sigma'(x+\epsilon)-\sigma'(x)}{\delta'(x+\epsilon)-\delta'(x)}=\frac{\sigma'(x)-\sigma'(y-\frac{a}{c}\epsilon)}{\delta'(x)-\delta'(y-\frac{a}{c}\epsilon)}$. By taking limit $\epsilon\to 0$, we obtain $\frac{\sigma''(x)}{\delta''(x)}=\frac{\sigma'(x)-\sigma'(y)}{\delta'(x)-\delta'(y)}$.\footnote{Technically, we need to write $\sigma''(x)(\delta'(x)-\delta'(y))=\delta''(x)(\sigma'(x)-\sigma'(y))$ since $\sigma''(x)=0$ is not ruled out. However, for ease of writing, we use the ratio expressions for the rest of this proof.} By setting $r=x, s=x+\epsilon$ and $t=y-\frac{b}{c}\epsilon$, we also obtain $\frac{\sigma'(x+\epsilon)-\sigma'(x)}{\delta'(x+\epsilon)-\delta'(x)}=\frac{\sigma'(x+\epsilon)-\sigma'(y-\frac{b}{c}\epsilon)}{\delta'(x+\epsilon)-\delta'(y-\frac{b}{c}\epsilon)}$. Hence, 

\[\frac{\sigma'(x)-\sigma'(y-\frac{a}{c}\epsilon)}{\delta'(x)-\delta'(y-\frac{a}{c}\epsilon)}=\frac{\sigma'(x+\epsilon)-\sigma'(x)}{\delta'(x+\epsilon)-\delta'(x)}=\frac{\sigma'(x+\epsilon)-\sigma'(y-\frac{b}{c}\epsilon)}{\delta'(x+\epsilon)-\delta'(y-\frac{b}{c}\epsilon)},\]
and the equality of the last and first terms implies 
\[\frac{\delta'(x+\epsilon)-\delta'(y-\frac{b}{c}\epsilon)-\delta'(x)+\delta'(y-\frac{a}{c}\epsilon)}{\delta'(x)-\delta'(y-\frac{a}{c}\epsilon)}=\frac{\sigma'(x+\epsilon)-\sigma'(y-\frac{b}{c}\epsilon)-\sigma'(x)+\sigma'(y-\frac{a}{c}\epsilon)}{\sigma'(x)-\sigma'(y-\frac{a}{c}\epsilon)}.\]
Then by taking the limit $\epsilon\to 0$ after dividing both sides by $\epsilon$, we obtain 
\[\frac{\delta''(x)+\frac{b-a}{c}\,\delta''(y)}{\delta'(x)-\delta'(y)}=\frac{\sigma''(x)+\frac{b-a}{c}\,\sigma''(y)}{\sigma'(x)-\sigma'(y)}.\]
Combined with the previous equality $\frac{\sigma''(x)}{\delta''(x)}=\frac{\sigma'(x)-\sigma'(y)}{\delta'(x)-\delta'(y)}$, since $a\neq b$, we obtain
\[\frac{\delta''(y)}{\delta'(x)-\delta'(y)}=\frac{\sigma''(y)}{\sigma'(x)-\sigma'(y)}.\]
Let $f(x)=\frac{\sigma'(x)-\sigma'(y)}{\delta'(x)-\delta'(y)}$. Then by differentiating the equality $f(x)(\delta'(x)-\delta'(y))=\sigma'(x)-\sigma'(y)$ with respect to $x$, we obtain \[f'(x)(\delta'(x)-\delta'(y))+f(x)(\delta''(x)+\delta''(y)(\frac{a+b}{c}))=\sigma''(x)+\sigma''(y)(\frac{a+b}{c}).\] Note that $f(x)=\frac{\sigma''(x)}{\delta''(x)}=\frac{\sigma''(y)}{\delta''(y)}$ implies $f(x)(\delta''(x)+\delta''(y)(\frac{a+b}{c}))=\sigma''(x)+\sigma''(y)(\frac{a+b}{c})$. Hence, $f'(x)(\delta'(x)-\delta'(y))=0$. Since $x\neq y$ for any $x\neq 1$, $f'(x)=0$ for any $x\neq 1$. Hence, $f(x)=A_{+}$ when $x>1$ and $f(x)=A_{-}$ when $x<1$ for some $A_{-}, A_{+}>0$. By the previous calculation, we have $\sigma''(x)=f(x)\delta''(x)$ and $\sigma''(y)=f(x)\delta''(y)$. The first equality by setting $x\in (0, 1)$ implies that $\sigma''(z)=A_{-}\delta''(z)$ for any $z\in (0, 1)$. Since $y=\frac{1-(a+b)x}{c}$, by setting $\frac{1}{a+b}>x>1$, the second equality implies $\sigma''(z)=A_{+}\delta''(z)$ for any $z\in (0, 1)$. Hence, $A_{-}=A_{+}$. 

Let $A=A_{-}=A_{+}$. Hence, for some $B$, $\sigma'(x)=A\,\delta'(x)+B$ for any $x<\frac{1}{a+b}$.

Moreover, because we have $A=\frac{\sigma'(x)-\sigma'(y)}{\delta'(x)-\delta'(y)}$ for any $x<\frac{1}{a+b}$, we also have $\sigma'(y)=A\,\delta'(y)+B$ for any $y<\frac{1}{c}$. 

Now recall that the equality $\frac{\sigma'(r)-\sigma'(t)}{\sigma'(s)-\sigma'(t)}=\frac{\delta'(r)-\delta'(t)}{\delta'(s)-\delta'(t)}$ is satisfied for any $r>0, s>0, t=\frac{1-a\,r-b\,s}{c}>0$. Note that we have $\sigma'(t)=A\,\delta'(t)+B$ since $t<\frac{1}{c}$. Moreover, for any $r<\frac{1}{a}$, we can find small enough $s$ such that $s<\frac{1}{c}$ (which implies that $\sigma'(s)=A\,\delta'(s)+B$) and $t=\frac{1-a\,r-b\,s}{c}>0$, the above equality implies $\sigma'(r)=A\,\delta'(r)+B$. By repeating the same argument for $s$, we obtain $\sigma'(s)=A\,\delta'(s)+B$ for any $s<\frac{1}{b}$. Therefore, $\sigma'(x)=A\,\delta'(x)+B$ for any $x<\frac{1}{\min\{a, b, c\}}$, which is the relevant domain. Consequently, $\sigma(x)=A\,\delta(x)+B\,x+C$ for some $C$. Hence, $d_\mu=A\,d'_\mu-B-C$.

The proof for $|S|\ge 4$ is identical to the above case by simply replacing $a=\mu(s_1)$, $b=\mu(E)$, and $c=\mu(S\setminus(E\cup \{s_1\}))$.\footnote{Note that there exists $E\subset S\setminus\{s_1\}$ such that $a\neq b$.}   

\subsection{Proof of Proposition \ref{prop:blackwell-agreement}}
For simplicity, $S=\{s_1, \ldots,s_n\}$ and recall that
$\mu,\pi,\pi'$ have full support. First, suppose $d_\mu(\pi) \leq d_\mu(\pi')$ for every Bayesian divergence. That is, for every strictly concave $\sigma$ with $\sigma(1)=0$, $\sum_{i} \mu_i \sigma(\frac{\pi_i}{\mu_i}) \geq \sum_{i} \mu_i \sigma(\frac{\pi'_i}{\mu_i})$. We will show there exists a row-stochastic matrix $K$ such that $\mu=\mu K$ and $\pi=\pi' K$. That is, we need a solution $K$ to $K_{ij}\geq0, \sum_j K_{ij} =1, \sum_i \mu_i K_{ij} =\mu_j, \sum_i \pi'_i K_{ij} =\pi_j$. By Farkas' lemma, this system has a solution or else there are numbers $a_i, b_i, c_i $ so that  $a_i+\mu_i b_j+\pi'_i c_j\geq0$ and $\sum_{i=1}^n (a_i + \mu_i b_i + \pi_i c_i) <0$. Suppose such numbers exist. Then  $b_j+c_j\frac{\pi'_i}{\mu_i} \geq-\frac{a_i}{\mu_i}$ for every $i,j$. Next, for $\varepsilon>0$, define $\sigma_{\varepsilon}(x)=-\varepsilon \ln \left( \sum_{j=1}^n e^{-\frac{b_j + c_jx}{\varepsilon}}\right)$, which is strictly concave provided not all $c_i$ are equal, which is ensured by the conditions of Farkas' lemma. Hence $\sigma_{\varepsilon}(x)$ satisfies the features of a Bayesian divergence. From $b_j+c_j\frac{\pi'_i}{\mu_i} \geq-\frac{a_i}{\mu_i}$, algebra yields $\sigma_\varepsilon\!\left(\frac{\pi'_i}{\mu_i}\right) \geq -\varepsilon\ln n-\frac{a_i}{\mu_i}.$ One can also show  $\sigma_\varepsilon\!\left(\frac{\pi_i}{\mu_i}\right) \leq b_i+c_i\frac{\pi_i}{\mu_i}.$
    Combining these, multiplying by $\mu_i$, summing and taking $\varepsilon$ sufficiently small yields 
    \[
    \sum_i\mu_i
    \left[
        \sigma_\varepsilon\!\left(\frac{\pi_i}{\mu_i}\right)
        -
        \sigma_\varepsilon\!\left(\frac{\pi'_i}{\mu_i}\right)
    \right]\leq
    \sum_i(a_i+\mu_i b_i+\pi_i c_i)
    +\varepsilon\ln n < 0.\]
The normalized function $\sigma_\varepsilon(x)-\sigma_\varepsilon(1)$ leaves the above difference fixed and is admissible for a Bayesian divergence, contradicting $d_\mu(\pi) \leq d_\mu(\pi')$ for every Bayesian divergence. Hence such a $K$ must exist. 

Conversely, suppose $K$ is row-stochastic and satisfies $\mu K=\mu$ and $\pi'K=\pi$. For each $j$, $\pi_j=\sum_i\pi'_iK_{ij}$. Dividing by $\mu_j>0$ and multiplying and dividing each
summand by $\mu_i>0$ gives 
\[\frac{\pi_j}{\mu_j}=\sum_i\frac{\pi'_iK_{ij}}{\mu_j}=\sum_i\frac{\mu_iK_{ij}}{\mu_j}\frac{\pi'_i}{\mu_i}.\] It follows that for any concave $\sigma$, 
\[\sigma\left(\frac{\pi_j}{\mu_j}\right)=\sigma\left(\sum_i\frac{\mu_iK_{ij}}{\mu_j}\frac{\pi'_i}{\mu_i}\right) \geq \sum_i\frac{\mu_iK_{ij}}{\mu_j}\sigma\left(\frac{\pi'_i}{\mu_i}\right) .\] Multiplying by $\mu_j$ and summing yields
\[\sum_j \mu_j \sigma\left(\frac{\pi_j}{\mu_j}\right) \geq \sum_j\sum_i\mu_iK_{ij}\sigma\left(\frac{\pi'_i}{\mu_i}\right) = \sum_j \mu_j \sigma\left(\frac{\pi'_j}{\mu_j}\right) .\] Since $\sigma$ was arbitrary, it follows that $d_\mu(\pi) \leq d_\mu(\pi')$ for any Bayesian divergence.

\subsection{Proof of Proposition \autoref{compstat}}

Let $\gamma^k=\mu^k_I(C)$. Then $\mu^k_I(A)=\mu^k_I(B)=\frac{1-\gamma^k}{2}$. By Proposition 1, it is enough to show that $\gamma^1<\gamma^2$. By way of contradiction, suppose $\gamma^1\ge \gamma^2$. By Proposition 1, $\gamma^k$ should maximize
\[\mu(C)\sigma^k\,\Big(\frac{\gamma^k}{\mu(C)}\Big)+\mu(A)\,\sigma^k\Big(\frac{\frac{1-\gamma^k}{2}}{\mu(A)}\Big)+\mu(B)\,\sigma^k\Big(\frac{\frac{1-\gamma^k}{2}}{\mu(B)}\Big).\]
The FOC for each $k\in\{1, 2\}$ gives
\[(\sigma^k)'\Big(\frac{\gamma^k}{\mu(C)}\Big)=\frac{(\sigma^k)'\Big(\frac{\frac{1-\gamma^k}{2}}{\mu(A)}\Big)+(\sigma^k)'\Big(\frac{\frac{1-\gamma^k}{2}}{\mu(B)}\Big)}{2},\]
which has an interior solution because $(\sigma^k)'$ is strictly decreasing. From the FOC for $k=1$, we obtain
\[f\Big((\sigma^1)'\Big(\frac{\gamma^1}{\mu(C)}\Big)\Big)=f\Big(\frac{(\sigma^1)'\Big(\frac{\frac{1-\gamma^1}{2}}{\mu(A)}\Big)+(\sigma^1)'\Big(\frac{\frac{1-\gamma^1}{2}}{\mu(B)}\Big)}{2}\Big).\]
Since $(\sigma^2)'=f\circ (\sigma^1)'$, we have \begin{eqnarray*}
(\sigma^2)'\Big(\frac{\gamma^1}{\mu(C)}\Big)&=&f\Big(\frac{(\sigma^1)'\Big(\frac{\frac{1-\gamma^1}{2}}{\mu(A)}\Big)+(\sigma^1)'\Big(\frac{\frac{1-\gamma^1}{2}}{\mu(B)}\Big)}{2}\Big)\\
&>&\frac{1}{2} f\Big((\sigma^1)'\Big(\frac{\frac{1-\gamma^1}{2}}{\mu(A)}\Big)\Big)+\frac{1}{2}f\Big((\sigma^1)'\Big(\frac{\frac{1-\gamma^1}{2}}{\mu(B)}\Big)\Big)\\
&=&\frac{1}{2} (\sigma^2)'\big(\frac{\frac{1-\gamma^1}{2}}{\mu(A)}\big)+\frac{1}{2}\,(\sigma^2)'\big(\frac{\frac{1-\gamma^1}{2}}{\mu(B)}\big)\\
&\ge& \frac{1}{2} (\sigma^2)'\Big(\frac{\frac{1-\gamma^2}{2}}{\mu(A)}\Big)+\frac{1}{2}\,\sigma^2)'\Big(\frac{\frac{1-\gamma^2}{2}}{\mu(B)}\Big)=(\sigma^2)'\Big(\frac{\gamma^2}{\mu(C)}\Big).
\end{eqnarray*}
The second line follows from the strict concavity of $f$ and $\mu(A)\neq \mu(B)$ (equivalently, $\mu\not\in I$), the third line is from $(\sigma^2)'=f\circ (\sigma^1)'$, and the fourth line is from the assumption $\gamma^1\ge \gamma^2$ and the fact that $(\sigma^2)'$ is decreasing. However, by comparing the first and last terms, we obtain $\gamma^2>\gamma^1$, a contradiction to the initial assumption.

\bibliographystyle{ecta}
\bibliography{econref_GI}

\end{document}